\documentclass{article}
\usepackage[totalwidth=480pt, totalheight=680pt]{geometry}
\usepackage{graphicx}
\usepackage{xcolor}

\usepackage{ifthen}
\usepackage[utf8]{inputenc}
\usepackage{amsmath,amssymb,amsthm}
\usepackage{color} %pour xfig
\usepackage{graphicx} % pour includegraphics appel\'e par xfig

\def\RR{{\mathbb R}}

\def\fxc{{\ensuremath \mathcal C}}

\newcommand{\fd}{\ensuremath{\rightarrow}}

\newcommand{\findem}{\nolinebreak\vspace{\baselineskip} \hfill\rule{2mm}{2mm}\\}
\renewcommand{\phi}{\ensuremath{\varphi}}

\renewcommand{\phi}{\ensuremath{\varphi}}

\newtheorem{nt}{Notation}

\newtheorem{prop}[nt]{Proposition}

\newtheorem{coro}[nt]{Corollary} 
\newtheorem{defi}[nt]{Definition}

\newtheorem{ex}[nt]{Example} 
\newtheorem{lm}[nt]{Lemma} 
 
\newtheorem{rem}[nt]{Remark} 
\newtheorem{thm}[nt]{Theorem}

\begin{document}
\sloppy
\title{Geometric approach for non pharmaceutical interventions in epidemiology}
\date{}
\author{Laurent Evain \footnote{laurent.evain@univ-angers.fr}, Jean-Jacques Loeb}
\maketitle

%%%%%%%%%%%%%%%%%%%%%%%%%%%%%%%%%%%%%%%%%%%%%%%%%%%%%%%%%%%%%%%%
%%%%                    Debut du contenu                  %%%%%
%%%%%%%%%%%%%%%%%%%%%%%%%%%%%%%%%%%%%%%%%%%%%%%%%%%%%%%%%%%%%%%%

\section*{Abstract: }

Various non pharmaceutical interventions have been settled to minimise the
burden of the COVID-19 outbreak.
We build a framework to analyse the dynamics of non pharmaceutical
interventions, to distinguish between mitigations measures leading to objective
scientific improvements and mitigations based on both political
and scientific considerations. 
We analyse two possible strategies within this framework. 
Namely, we consider mitigations driven by the limited resources of the
health system and mitigations where a constant set of measures
is applied at  different moments. We describe the optimal interventions for these scenarios.  
Our approach involves sir differential systems, it is qualitative and geometrical rather than computational.
Along with the analysis of these scenarios,
we collect several results that may be useful on their own,
in particular on the ground when the variables are not known in real
time.

\section{Introduction}

In the pandemic situation, governments have settled policies, based
on socio-economic appreciations, field studies,
and modelling.  The toolbox 
for the crisis management
involved mitigations policies.
Numerical simulations suggest that these mitigations measures
change the final share $r_\infty$ of infected people in the
population, sometimes markedly.
A possible roadmap for a scientific programme to select
non pharmaceutical interventions could
be as follows. 
\begin{enumerate}
\item Discuss the choice of the model. Which models are realistic ?
\item Find the mathematical
  optimisations for the chosen model. 
\item Find the counterpart of the optimisation in real life. 
  Determine possible concrete mitigations 
  that approach the targeted mathematical optimisation. 
\item Social analysis. Analyse the public acceptance of the
  mitigation, the economic impact, and the indirect costs on the
  population health. 
\end{enumerate}
In the present article, we choose a variant of  the sir-model. We
concentrate on the second item of this
roadmap. Our target is to obtain theoretical results, and in
particular proved qualitative results. Our
results shed light and give an understanding of the  
the numerical simulations of the spread of a disease.
We have a particular interest in qualitative
results independent of the input parameters, as these results could be more
robust on field with little known parameters. The  theoretical
background, the constructions, and the mathematical results are exposed in
full generality in the supplementary material. 
The main text is dedicated to a larger audience, it
explains, contextualises and illustrates  the results with
simulations.

Our work started with the article \cite{britton}. We were puzzled by some 
simulations exhibiting final fractions infected  depending on the
choice of the intensity of
preventive measures, via a constant $\alpha$ in the next generation
matrix. Several remarks were formulated, suggesting
qualitative explanations of the phenomena observed on the simulations.
For instance, ``preventive measures were not imposed from
the start and were lifted before the epidemic was over'' or
%``this
%epidemic was further from completion when sanctions were lifted''
%or
``lifting restrictions gradually [ can prevent ]
overshoot ``. So
it was implicit but clear from the article \cite{britton} 
that an adequate scheduling was required to minimise the burden of
the epidemic. However, we could
not identify the adequate scheduling in precise
mathematical terms, nor could we identify
the assumptions required
to obtain the qualitative behaviour of the examples. Thus our  goal
was  to clarify and give a general picture of what could mean an ``optimal scheduling''
of the preventive measures for a pandemic outbreak
described  with a
sir-model.

Our interest for qualitative results was reinforced by contradictory
results in the literature with respect to the relevance of an early
and strong set of mitigation measures. Whereas overshoot was pointed out as a
risk in  \cite{britton} and implicitly in \cite{ferguson}, other
other sources \cite{sofonea}, \cite{carcione-ducrot}
advocated for early or strong mitigations
to save lives. 

In contrast to papers where strategies are analysed at fixed dates
\cite{sofonea}, sometimes to wait for some new drugs or vaccines
\cite{ferguson}, we are concerned with the very long term. We try to
minimise the mortality after an infinitely long time using only non pharmaceutical interventions.
In other words, this paper considers non pharmaceutical interventions
as an active medical tool to minimise the burden of the epidemic in the long
run rather than as a tool to postpone the mortality till some new drug
comes on the market. 

Here is a summary of our results.
\begin{itemize}
  \item Even in the absence of medicine to wait for, finite time interventions may be  
  considered to minimise the burden of an epidemic because
  of the dynamics involved. The situation is
  analogous to a bike on a sloping road : it is not possible to stop
  before the low point using a finite time breaking, but
  breaking is nevertheless useful to avoid moving far beyond 
  the low point due to inertia. In symbols, let  $r_{\infty}$ be the ratio of finally infected people 
and let $R_0$ be
  the classical reproduction number. 
 A well scheduled finite time intervention can
  drive $r_{\infty}$ close to
  $r_{herd}=1-\frac{1}{R_0}$, whereas
  no intervention often leads to 
  $r_{\infty}>>r_{herd}$. The role of the
  dynamics is more important when $R_0$
  has a medium value ($R_0\simeq 2.5)$  where nearly 30\% of the population may avoid
  the disease thanks to a suitable finite time intervention.
\item There is a fundamental qualitative difference between finite
  time interventions and infinite time interventions. Infinite time
  interventions can lead to an arbitrarily small ratio $r_{\infty}$ of
  infected people. In
  contrast,  finite
  time interventions result in situations where the inequality
  $r_{\infty}>r_{herd}$ always holds, so that  $r_\infty$ close to $r_{herd}$ is the best
  possible value. 
\item Planning is important. Examples show that awkward planning lead to mitigations
  that are long, costly, with little effect on $r_{\infty}$. The
  analogy with  bikes on a sloping road makes sense again : breaking hard far from the
  low point hardly has an impact on the inertia and on the distance
  covered after the low point.  In contrast, the analogy with a bike
  on a flat road is badly suggestive, as it encourages early
  intensive mitigations with poor results. This
  leads to a problem of control theory : what are the mitigations
  which minimise the effort on the population for a fixed result ? 
\item We build a scientific framework
to distinguish between the
political level and the scientific level in the decision process. 
Obviously, the mitigations have a social cost that require a personal
subjective appreciation. A rational decision process
includes political considerations to aggregate the divergent
wishes of the citizen. Nevertheless, some conclusions
may be true independently of the subjectivity. Considering two possible
choices $A$ and $B$, there is a scientific ground to prefer
mitigation $A$ to mitigation $B$ if
there are simultaneously fewer infected people and fewer restrictions on the
population when $A$ is chosen. In contrast, a political trade-off is
necessary when a middle ground between infections and constraints
has to be found. We keep these notions informal in the main text, but
the concepts are rigorously defined
in the supplementary material in terms of cost functions and
diffeomorphisms.  In the following, when we say that a choice $A$ is
better than a choice $B$, we always refer to the scientific meaning :  
choice $A$ leads both to less
restrictions  and to less infected people than choice $B$. 
\item We analyse the scenario where a same
  constant intervention is applied one or several times.
 In this context, the two problems of minimising the duration of the
 constraint for a ﬁxed burden or
minimising the burden for a ﬁxed duration are equivalent, and there is
no compromise between the duration of the intervention and the number
of infections to be found : both are
minimised simultaneously. We thus approach the problem
with a fixed predefined duration and we determine the adequate
planning. This
 scenario includes
 for instance the comparison between
 two strategies, where the first strategy promotes a change every Monday for seven
weeks, whereas the second strategy promotes the same change a whole week one month after
the starting point of the epidemic.
More generally, we compare strategies with a same type of
intervention, and the same total duration, and we analyse the optimal
planning of the mitigations. 
We show that in this scenario, splitting the mitigations through
several short periods is never optimal.
The mitigation minimising the
number of finally infected people has always exactly 
one unique long lasting mitigation.
Our model does not support
the idea sometimes expressed on the media
to plan a strong mitigation as soon as possible. It is quite the
opposite. Intensity and timing have to
be tuned in a consistent balanced manner : 
an earlier mitigation must be lighter than an
intervention that starts later. Early and strong mitigations are not
balanced and yield to poor results. The timing of an optimal
planning is understood : it
boils down to “as soon as possible” if the herd immunity threshold has been crossed,
and around the herd immunity threshold otherwise.
Among the possible consistent choices
of timing and intensity, the case of a late intensive
mitigation, starting when the herd immunity threshold is crossed,  has a
special interest. The corresponding strategy can be implemented
on the ground using measurements in wastewater as in \cite{obepine1}.
It may be more easily planned than alternative optimal strategies
since estimating $R_0$ is
not necessary, thus bypassing the
diﬃculty of its  estimate. 
  \item In a second scenario, we analyse the case where the health system would be saturated in
    the absence of interventions. Non pharmaceutical interventions are used to maintain
    the health system below its maximal load. 
    We compare several mitigation strategies with
    different loads, possibly different from 100\%. 
    For instance, the mitigation may start when 
the health system is filled at 90\%. This second scenario is
    divided in two sub-scenarios :
    \begin{itemize}
    \item 2a) : The mitigation is shaped so that the health system
      stays filled at 90\% and it is relaxed when the herd immunity
      threshold is reached. The relaxation occurs when the epidemic
      naturally decreases. 
    \item 2b) :  This scenario starts like scenario 2a), but the mitigation is
    relaxed sooner. A
    rebound of the epidemic occurs and the limit of 90\% is exceeded after relaxing, but the total load
    remains below 100\% forever. In other words, the relaxation is launched as soon as
returning to normal does not overload the health system in the future despite
    of a rebound.
  \end{itemize}
  Instead of considering arbitrarily a load of
 90\% as in the above example, we address the problem of determining the optimal load between $0\%$
 and $100\%$ for the health system. What are the optimal loads for scenarios 2a) and 2b) ?  
    First, we show that in scenario 2a), there is no scientific
    answer. Political trade-offs are unavoidable :
    a higher load abuts to fewer infected people at the price of more
    constraints.  In scenario 2a),  the duration of the mitigation
    tends quickly to infinity when the considered
    load goes to zero. Simulations show that 
    the time of mitigation is often very large.
    It is thus natural to consider scenario 2b)
    which comes with fewer constraints. 
    We show that, maybe surprisingly, all the strategies considered in
    the scenario 2b)
    have the same number of finally infected people, independently of
    the chosen load. Consequently, a load
    $A$ is preferable than a load $B$ if and only if the corresponding
    strategy leads to fewer
    constraints on the population for the same result. 
    Small loads are inefficient in scenario 2b). A minimal load  is necessary,
    otherwise the policy is surpassed by other
    better planned strategies. In the simulations considered, this minimal
    load of the health system to reach before launching the intervention is large : more than
    80\% of the maximal load of the health system. 
    This may be viewed as an other incarnation in the context of
    possibly overloaded health systems of the slogan ``A 
    strong and early mitigation is inefficient''. When
    this minimal load is reached, the question of still enlarging the
    launching load becomes a political one, it is not a scientific question
    any more : for the same number of finally infected people,
    a higher launching load requires an effort for the population
    which lasts longer, but the maximal effort 
    is lower.
  \item The above scenarios are built upon a more general
    analysis which
    carries several results
    useful on their own. We 
    give a  focus on 
  a function $h$  which plays a role similar to energy
  in physics. In mechanics, a falling object undergoes important
  damages on
  the ground if it was thrown
  with a high kinetic or potential energy. Similarly in our model, if
  a mitigation is relaxed with a high value of $h$, this will lead to many infected
  people and fatal cases because of the implied dynamics. Since $h$ is
  an indirect measure of the finally infected people,  a mitigation that lowers $h$ 
  has a positive impact on the final burden. In contrast,  if $h$ is
  only slightly changed by a mitigation, the mitigation hardly has an
  impact on the final burden : The mitigation is  a temporal shift rather than
  an amelioration, the infections will happen later. 
  A sensible objective for the public policy is thus to
  lower $h$  using interventions of short duration,  hence the importance of 
  the derivative  $\frac{dh}{dt} $. The computation shows that
  $\frac{dh}{dt} $ is proportional to the ratio $i(t)$ of
    infected people for a fixed intervention. This leads to the
    very important qualitative result that  for a fixed 
    level of constraint, the interventions are more
    efficient if they occur when many people are
    infected. The same phenomena has been observed using numerical
    simulations in \cite{ferguson}, however for a different
    model. This  suggests that our qualitative result for the sir
    model could be extended and could provide an explanation of the numerical
    observations in other contexts.  The variation of the
    energy function $h$ shows that an early intensive intervention
    acts on mortality similarly
    to a free loan, with a positive effect on the short term but not
    on the long term after the loan is repaid. This
    explains the apparent contradiction between the authors studying
    the aftermath of the intervention at a fixed date with a positive result \cite{sofonea},
    whereas \cite{britton},\cite{ferguson} 
    have an opposite conclusion with a further time horizon. 
\end{itemize}

\subsection*{Context and limitations of the findings}
\label{sec:cont-limit-find}

Mitigation strategies
may be promoted through coercive laws, or
they may be exposed to the public as mere
recommendations with no obligation. Citizen may
have more choice when the intervention occurs
( new option for remote work or for a day off for instance) or fewer choices due to restrictions. 
Our paper is agnostic about the implementations.
We consider a model which modifies the coefficients of
the sir systems when people modulate their interactions,
but we make no assumption on the social tools used
to modify these coefficients.

Several questions have not been considered.

We do not discuss the economical or global health impact  nor the social acceptance
in this paper.

Field testing, 
studies with animal models and experiments to
support the numerical and theoretical results would be welcome.
Insect pathogens have been used
to test equations because of their tractability \cite{dwyer,dwyer2}.
Minimising the spread of the epidemic in plants  while
minimising the intervention is a natural question,
and we have not explored how our results could be enlightening for agriculture
or other epidemiological contexts.

These questions are out of the scope of the
article, and they belong to a field of work where we have no expertise. 
However, we believe that these are  important questions to be
discussed by qualified researchers  in these fields.

% Given two
% strategies, strategy S1 is preferable on an objective
% scientifc reasons than strategy S2 if $S_1<S_2$ for the partial
% order. If politiical decisions are necessary then $S_1$ and $S_2$ are
% not comparable inside the framework of the partial order. Intuitivly,
% $S_1$ is preferable as suggested above if both the number of infected
% people and the constraint are better in $S_1$. In rigorous
% mathematical terms in the supplementary appendix, we introduce cost
% functions, and some specific diffeomorphisms to check this partial
% order. 

We discuss now the choice of the model.
Some modellings rely on simple models, whereas other modellings
require many interacting parameters. Both have their
pros and cons, depending on the
objectives, quantitative or qualitative, and on the quality of the
data. As a rule of
thumb, simple models with few variables 
allow qualitative explanations, they have a lower sensitivity to the
parameters.  When high quality data and model is available, complex models with 
more variables lead to more precise predictions. 
Qualitative interpretation
of the changes implied by the modifications of the input constants 
is difficult for complex models. Both
approaches are complementary rather than opposite.

% Tom Britton's paper \cite{} is rather
% qualitative; it clarifies the change in the constant $R_0$ caused by
% heterogeneity in the population. The model by
% Ferguson and  al. is more elaborated and more
% precise on the constants. for this sophisticated model.  Qualitative models must be
% realistic enough to provide sensible orders of magnitude in
% the simulations, quantitative models should control the noise,
% which may be large because of the numerous variables. 

Since our goal was to 
identify the qualitative phenomena  that drive and circumscribe the
computations, a variant of
the simple and robust sir-model was a sound choice. In the variant considered,
a coefficient that was constant in the original sir model varies with
the mitigation policies implemented.  
Although the sir model is suitable
for qualitative analysis, our modelling carries
the simplifications and limitations attached to this model :
people are infected only once, deaths are not considered, the
population is supposed to be geographically homogeneous,
all individuals are equally susceptible, viruses undergo no mutations,
to cite a few limitations.

There are slight variations of the sir-model for which we can
carry our results or follow an approach along the
same lines ( remark \ref{rem:otherModels}). 
However, there are other models to
experiment to approach reality, and they may be
quite different.  For instance, ``on the experimental side, Dwyer et al. (1997) measured
nonlinear relationships between transmission and densities of
susceptible hosts, implying that the bilinear term in the classical
susceptible-infected-recovered (SIR) model may not be appropriate.  ``
\cite{gomes,dwyer}. Many variations are possible, and probably many
are necessary depending on the problem under consideration.
Closed formulas for a differential system are an  exception.
Most variations from the SIR-system will lead to models
which are not computable with closed formulas. 
While a large part of our analysis is geometrical and clarifies
the involved phenomenons,  other arguments still depend
on the closed formulas of the sir system and will not apply to non
computable models. Generalisations 
based on a better understanding of the geometrical architecture
could be explored,  to prove our results
for a larger class of models.

Many parameters are not constant, their value evolves with time. For instance,
we do not know how many
people will be vaccinated, how often, and the efficiency of the vaccines on the
variants to come. 
In this rapidly changing environment, 
we hope that our qualitative results
may be  useful, in particular those
independent
of the numerical data in input. 
% However,  the qualitative
% analysis remains the same, as long as the efficiency is constant
% in the period considered for  the analysis. 

\section{Main text}

\subsection*{Mitigation as an active tool}
In a pandemic context, mitigation policies are usually understood as
a tool to give deciders and physicians time for dealing with the
problem. Delaying the epidemic gives time for researchers to find new
remedies and gives time to setup the logistics to vaccinate people. 
Our approach in this paper is different. We consider mitigations as a
tool to minimise the number of finally infected people, even in the
absence of change in the remedies. The goal of this article is to
develop this point of view of mitigation strategies as an active tool
for driving the epidemic. In this section, we exhibit 
simulations to expose the problematic 
involved.

A mitigation that improves the final
situation without  any new drug is illustrated in the next figure.
The red drawing is the trajectory of the pandemic without mitigation.
The green drawing is the same pandemic, where a mitigation is applied
from week 22 to week 26. During  the four weeks of mitigation, the
reproduction number $R_0$ has been reduced from $2.4$ in the absence
of mitigation to $R_1=0.4$.
Thanks to the mitigation, the share $r_\infty$ of finally infected
people dropped from 0.88 to 0.82. 

The evolution of the pandemic  is shown
in the $(s,r)$-plane, more precisely in the triangle $s> 0,r\geq 0,
s+r\leq 1$. Here $s$ and $r$ are  the share of susceptible
people and the share of removed people in the population, with the standard
notations of the sir model.  The share $i$ of infected people 
is implicit since $i=1-r-s$.
The point $M(t)=(s(t),r(t))$ that represents the epidemic at time $t$ 
starts at $t=0$ with $M(0)$ at the bottom  right ($M(0)\simeq
(1,0)$). 
As $t$ increases, $M(t)$ moves to the left ($s$ decreases) and the
limit point
$M_{\infty}=(s_\infty,r_\infty)$ is on the diagonal  $r+s=1$ ($i_\infty=0$).
Note that  by construction $r_{\infty}$ is
the share of people  that have been infected
at some time $t$ during the epidemic, with $0<t<\infty$.
The limit point $M_\infty$ is lower  for the
green curve than for the red curve.
This means means that the mitigation has been efficient,
it has lowered $r_{\infty}$. 

\begin{figure}[h]
\centering
\includegraphics[scale=.40]{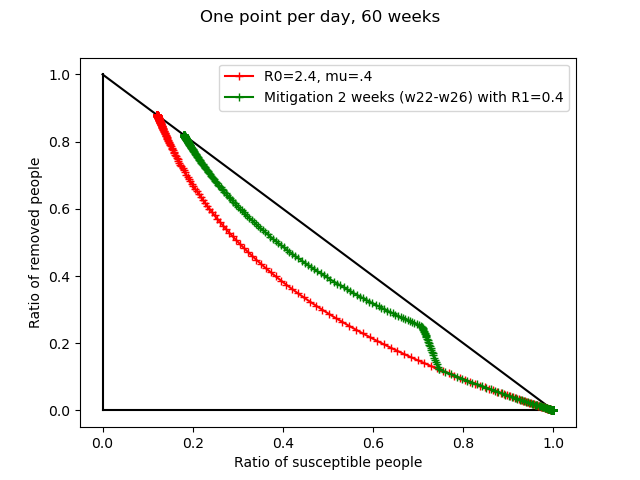}
\caption{A simple mitigation}
\label{simpleMitigation}
\end{figure}

What is a good mitigation strategy ? The problem can be thought in
analogy with the  braking of a bike on
a slope. All braking
strategies are not equivalent : a rider  does not
apply a constant braking force on downslopes. 
There are moments where braking is useless, and  other moments where
braking is necessary to take the turn.  Technically, this is a
question of control theory. The riders on the Tour de France
unconsciously apply  some control theory to find the optimal
timing and intensity for the braking.

The same phenomenon appears in an epidemiological context.
A mitigation measure is the analogue
of a braking action to slow down the epidemic. Since mitigations limit
the possibilities for citizen, one wants to minimise the duration and intensity
of these mitigations for the same braking performance, or to minimise the number
of infected people for a fixed level of braking effort.

A na\"\i ve approach would suppose that
control theory is straightforward, that no planning is required, and that
the value of $r_{\infty}$ depends only on how strict and how long the
mitigations are. 
This is not the case :  adequate scheduling is important. A lighter and shorter
mitigation may outperform a harsher mitigation  thanks to a better
scheduling, as illustrated with the following example. The
mitigations for the green trajectory are shorter-lived than those for the
red trajectory ( 4 weeks vs 6 weeks ), are less intensive
( reproduction during the mitigation R1=0.8 vs R1=0.7), yet the share $r_{\infty}$  of finally infected people is
lower for the green curve. A public policy should recommend the green
trajectory over the red trajectory. 
\begin{figure}[h]
\centering
\includegraphics[scale=.40]{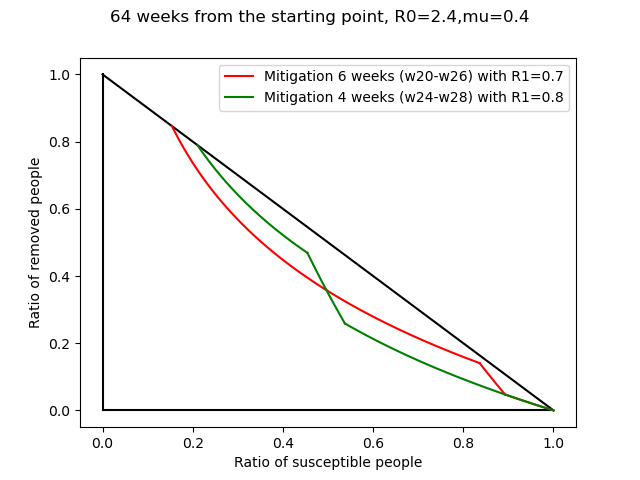}
\caption{Two mitigations with different plannings.}
\label{2mitigations}
\end{figure}

\section{Orders of magnitude}

What is the difference between a perfect mitigation and no mitigation
? How many  saved lives and how many people may avoid infection using an efficient mitigation
? We give some estimates in this section.
% In the first simulation, the epidemic starts at time
% $t=0$ with reproduction number $R_0$ and mitigations measures
% reduce this number to $R_1$.  Every mitigation lasts 2 months but is
% planned at different moments.  The population comes back to normal
% life when the mitigation is over. The table shows the share $r_{\infty}$ of finally
% infected people in the whole process, \textit{i.e.} the share of people that have been infected
% at some time $t$ during the epidemic, with $0<t<\infty$.

In a free environment without mitigation, the share $r_{\infty,free}$ of finally
infected people in a sir model is the unique positive solution of the equation
\begin{displaymath}
  \ln(1-r_{\infty,free})+R_0r_{\infty,free}=0.
\end{displaymath}
( Theorem \ref{thm:damageIncreasesWithEnergy}). With an optimal mitigation, the ratio of finally infected people is
about  
\begin{displaymath}
r_{\infty,opt}=1-\frac{1}{R_0}
\end{displaymath}
(
Theorem \ref{thm:optimum_absolu_pour_une_strategie}).
The share $\Delta_r$ of the population avoiding an infection with an
optimal mitigation is thus
$$\Delta_r= r_{\infty,free}-r_{\infty,opt}.$$
The share
$\Delta_{deaths}$ of lives saved is
$$ \    \Delta_{death}=(r_{\infty,free}-r_{\infty,opt})IFR. $$
where IFR denotes the infection fatality rate.
\\
%We shall prove later theorical results stating that the infimum share
%of finally infected people is $r_{herd}=1-\frac{1}{R_0} $ and the share
%$r_{\infty}$ in the absence of mitigation can be derived from
%equation ....  The
%number $\Delta$ of peple whose death could have been avoided through an
%adequate mitigation strategy is then 

  \begin{figure}[h]
    \centering
  \input{tableau2}
  \caption{Orders of magnitude}
      \label{orders_of_magnitude}

   \end{figure}

Some estimates of these quantities are  given in the table of figure \ref{orders_of_magnitude} 
for different values of $R_0$ and $IFR$. When $R_0$ is slightly above 2 : up to 30\% of
the general population avoids an infection with an optimal
mitigation.  This high figure means that, from a mathematical point of
view, mitigation as a tool to reduce the burden of an epidemic makes
sense. There is no guarantee however that this strategy is possible in real life.

We remarked quite surprisingly that the importance of mitigation
increases for medium $R_0$. The natural guess that mitigation should
play a
more important role for $R_0$ large is wrong : for large $R_0$, the difference
$\Delta_r=(r_{\infty,free}-r_{\infty,opt})$ is small because $0<r_{\infty,opt}<r_{\infty,free}<1$ and $r_{\infty,opt}=1-\frac{1}{R_0}$
is close to $1$. This phenomenon is illustrated in figure
\ref{DeltarFunctionOfR0}
showing $\Delta_r$ as a function of $R_0$. 
\begin{figure}[h]
\centering
\includegraphics[scale=.30]{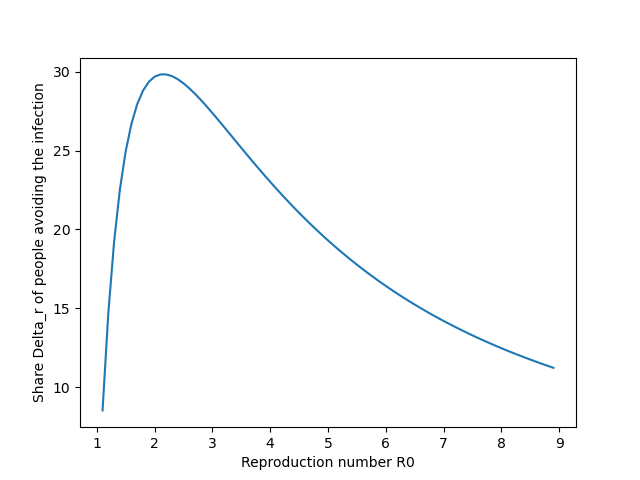}
\caption{$\Delta_r$ as a function of $R_0$ }
\label{DeltarFunctionOfR0}
\end{figure}

Note that the figures in table \ref{orders_of_magnitude} for
$\Delta_r$ are only upper-bounds for
the objectives of the public
policies. People often react naturally
when a brother or a friend is
ill. In a growing epidemic, the population limits its
interactions by itself. For instance, it was remarked in
\cite{r0_used_in_owid}
that ``most of the decline in mobility in
[the]  sample happened before the introduction of lockdowns.
Failing to account for voluntary
changes in behaviour leads to substantially over-estimated effects of
non pharmaceutical interventions
''.
We call $\Delta_{nat}$ the share of the population
avoiding an infection  due to this reaction
of the public. The share 
$\Delta_{public\ policy}$ of the population protected
against infection by the public
policies is in addition to $\Delta_{nat}$. 
The population
that avoids an infection thanks to mitigation
is $\Delta_{nat}+\Delta_{public\ policy}$.
If the combined effect of natural reaction and  public
policies is optimal, $\Delta_{nat}+\Delta_{public\ policy} =
\Delta_r$. Without the optimality hypothesis,
$\Delta_{public\ policy}\leq \Delta_r-\Delta_{nat}$. 
In our views, the role of
the political institutions is to coordinate and amplify if necessary the natural movement of
the public to
maximise $\Delta_{public\ policy}$, in accordance to the history and
social context of the country. Planning and scheduling the information to the public is an
ingredient of this maximisation, as media
campaigns can increase the level of compliance to
safe attitudes at the appropriate time.  Prophylactic measures
are dependent on the mode of transmission of the disease. Promoting
the correct prophylactic attitudes at the key moments could also be
an objective of the public policy.

\subsection*{Studying scenarios}
\label{sec:studying-scenarii}

Our next goal is to give a qualitative analysis of the involved
phenomena during the mitigations. To this aim,
we tried to go beyond numerical simulations, because their qualitative interpretation
is often difficult, and because extracting a general
behaviour from examples depending on the input data is in our opinion a slippery
methodology. Rather, we consider two scenarios, and we prove qualitative
results that apply independently of the numerical
data in input. The first scenario has  fixed mitigations.
The second scenario considers situations where the health
system is overwhelmed in the absence of mitigation. 

In the first scenario, we consider a fixed action : for instance a
larger part of the population 
works remotely. This leads to a mitigation with reproduction number
$R_1$ which is lower than the initial reproduction number
$R_0>1$. Both cases $R_1>1$ and $R_1<1$ make sense. We fix a total
duration $d$
for the mitigation measure. This whole mitigation  is 
split in $k+1$ 
shorter uninterrupted mitigations of duration $d_0,\dots,
d_k$.  The total duration of the mitigation is the sum of the
duration of the uninterrupted  mitigations, hence $\sum d_i=d$.
In this context, the question is : what is the
optimal value for the number $k$ 
and when should these $k+1$ mitigations occur ?

Our answer
is that the optimal strategy satisfies $k=0$. The optimal strategy
is an uninterrupted mitigation, which is not split in several shorter
mitigations. This is a general fact independent of the values of
$R_0$ and $R_1$. It is illustrated in  the left part of figure
\begin{figure}[h]
\centering
\includegraphics[scale=.4]{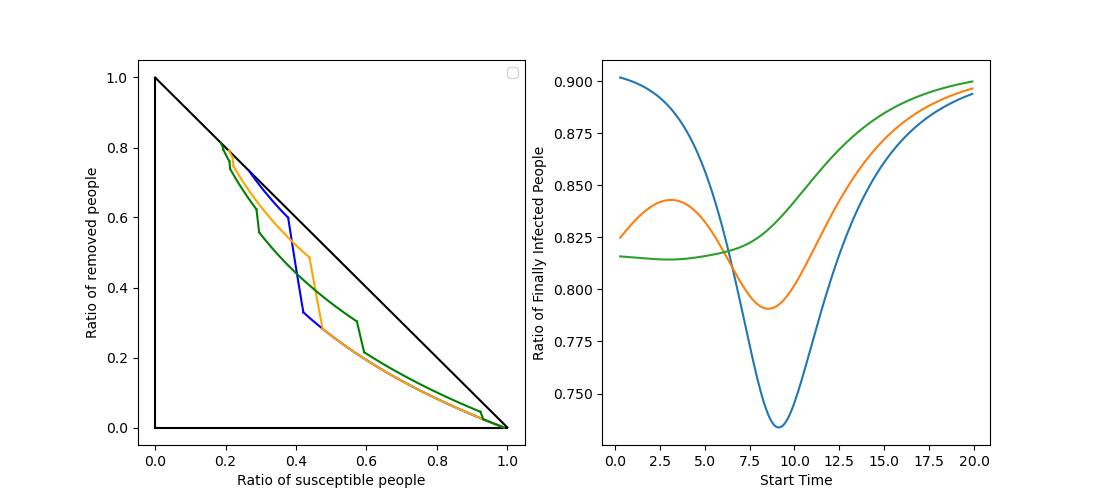}
\caption{Splitting mitigations}
\label{Splitting_mitigations}
\end{figure}
\ref{Splitting_mitigations} : a 60 days mitigation lowers the
reproduction from $R_0=2.6$ to
$R_1=0.4$ ( $\mu=0.4$). This 60 days mitigation is not split 
(blue curve, partially hidden by the orange curve), split in 
two shorter mitigations of 30 days (orange curve), or five shorter mitigations
of 12 days (green curve). The pause between two successive mitigations is three times the duration of
the mitigations, namely 90 days ( orange curve) and 36 days ( green curve). For each of these three
scenarios, the start time for the first
mitigation has been chosen
to give the smallest possible $r_{\infty}$. If the start time of the
first mitigation is
changed,  the value of $r_{\infty}$
depending on the start time (in weeks ) is given in the right part of
the figure for each scenario.  For instance, for two mitigations of 30
days distant of 90 days, the orange curve on the right part of figure \ref{Splitting_mitigations}
says that in our simulations,
the optimal start time for the first mitigation is a little more than 8 weeks after
the first few imported cases, and yields to $r_\infty$ between $0.775$ and
$0.8$. The corresponding epidemic is drawn on the left side.

Moreover, in the unsplit case $k=0$, we have an estimate
for the optimal moment for the intervention. Optimality requires that 
the equality $s(t)=s_{herd}=\frac{1}{R_0} $ occurs 
at a moment $t$ during the mitigation. This is also a general result
independent of the numerical values of $R_0$ and $R_1$. 
If the mitigation ends at a time $t$ with $s(t)<s_{herd}$, it is too
early. If the mitigation starts with $s(t)>s_{herd}$, it is too late
( Theorem \ref{thm:optimal_constant_mitigations}).
In the
example of the figure, with $R_0=2.6$, we have $s_{herd}=0.38$.
The optimal strategy illustrated by the blue curve of the figure has
numerical results consistent with the general theorem :  the vertical line $s=s_{herd}=0.38$ is
crossed during the mitigation period, represented by the most vertical
part of the left blue curve. 

Using the analogy with a bike on a slope, and considering that the low
point on the road is the analogue of the herd ratio, our theorem says that the optimal breaking occurs in
one step, and we should be breaking in a zone which encompasses
the bottom of the slope. Ending
the breaking before the lowest point of the valley would be
inefficient, starting the breaking after the low point would be
inefficient too.

\medskip

In the scenario considered so far, the constraint was the same for all
strategies ( same restrictions, same total duration for the mitigation).
In the next scenario, it will be necessary to compare heterogeneous
constraints, which are not constant in time nor have the same
duration. The comparison is easy in some cases. For instance, 
it is less constraining to have a soft mitigation lasting two days than
a harsh mitigation lasting five days. For some other cases,
the comparison is not possible : There is no natural
choice between a long soft constraint, and a short harsh constraint.
Finally, there are comparisons which are possible
but may require a moment to reflect.
As an example, a strategy
$S_1$ imposing a partial set of constraints for 2 days and a
total set of constraints for 3 days  is less constrained than an
alternative strategy $S_2$ imposing the same partial mitigation for $1$ day, and 
the same total mitigation for 5 days ( reason: $S_2$ is obtained from $S_1$ by
replacing one day of partial constraints with two days of complete
constraints).
This approach to order the constraints on the above examples can be formalised and
written rigorously.  In the appendix, we formalise
and extend these ideas to compare the constraints
of two different  strategies $S_1$ and $S_2$ to the
case of mitigations whose constraints vary continuously
with time.  To keep things simple yet intuitive, 
there are fewer
constraints for the strategy $S_1$ than for the strategy $S_2$ if
every person prefers $S_1$ than $S_2$, whatever her
personal cost function. This occurs in particular when $S_2$ is
obtained from $S_1$ by replacing mitigations of duration $d$
with harsher mitigations of duration $d'>d$.

In the second scenario, we consider a risk of overwhelmed health
systems. We fix an upper bound $i_{trig}$ for the maximal ratio of
infected people. For instance,
$i_{trig}=0.1$ means that  when 10\% of people are infected simultaneously, a mitigation
is triggered to prevent the increase of the number of hospitalised patients. The
level of mitigation is then settled so that
the ratio of infected people stays exactly at $i=i_{trig}=0.1$. 
Equivalently, there are as many people getting sick as people
being cured by
unit of time during the mitigation period. The process is
illustrated in the next figure. The epidemic starts with very few infected
persons, and the initial propagation is drawn in red. 
Then, the level of infection which triggers the mitigation measures is
reached and the mitigation period with constant $i$  
corresponds to the blue segment. After some time
(13 weeks in the example of the figure, $\mu=0.33$ ), the level of infected people naturally
decreases. The mitigations are relaxed as they are not necessary
any more ( in orange on the figure).  A similar strategy can be
set up with an other value for $i_{trig}$.
The question in this context is the determination of the
optimal $i_{trig}$ ? 
\begin{figure}[h]
\centering
\includegraphics[scale=.3]{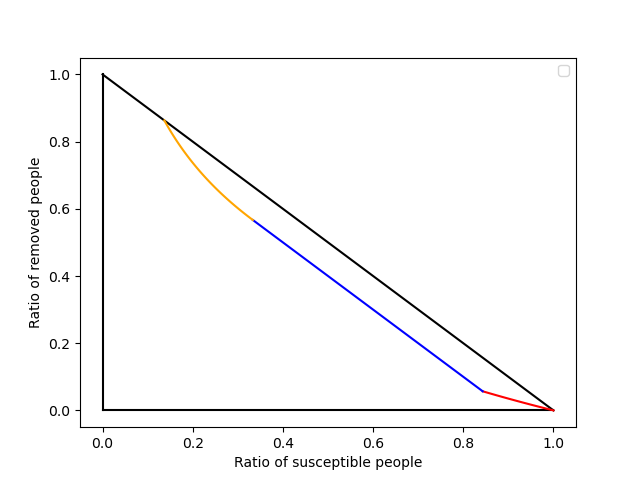}
\caption{Fixing the maximum level of infected people }
\label{SansRebondUnique}
\end{figure}

We show that there is no scientific answer to this question.
A political trade-off is necessary, as minimising
constraints and  minimising the ratio $r_{\infty}$ of
finally infected people are opposite objectives.  A small $i_{trig}$
corresponds to a smaller $r_{\infty}$  at the price of more
constraints (Theorem \ref{thm:witoutRebound}).
\begin{figure}[h]
\centering
\includegraphics[scale=.3]{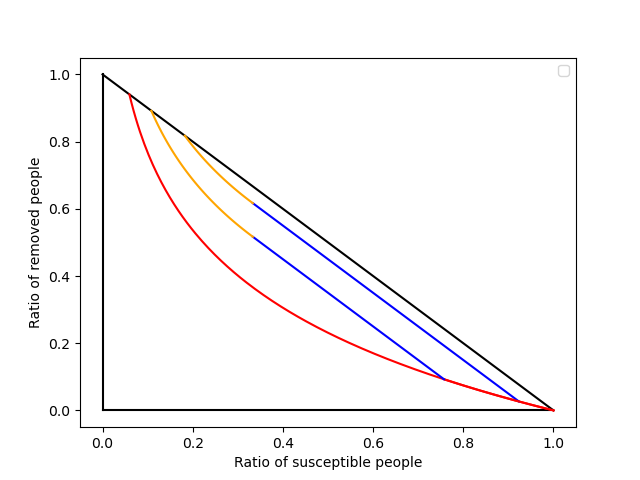}
\caption{Comparing two different values  of $i_{trig}$.}
\label{sansRebond}
\end{figure}
Figure \ref{sansRebond} illustrates this fact with
two mitigation levels $i_{trig}=5\%$  and $i_{trig}=15\%$. 
The red curve represents an epidemic without mitigation.
The two possible mitigations are drawn in
blue. The ratio $r_{\infty}$ is smaller
for $i_{trig}=5\%$, but the mitigation lasts far longer than for
$i_{trig}=15\%$ ( 35 weeks vs 8
weeks with $R_0=3$, $\mu=0.33$).  Moreover, the initial reproduction number of
the mitigation ( defined as the reproduction number when the
mitigation has just started)
is smaller  for
$i_{trig}=5\%$ ( 1.08 vs 1.32 ). This means  that harsher and longer
constraints are necessary for a small $i_{trig}$ value.

As $i_{trig}$ tends to $0$, the constraint duration tends rapidly to
$\infty$. This is 
illustrated in figure  \ref{timeVsIma} where the duration of the
mitigation in weeks is plotted in red, and the initial reproduction number
defined above is plotted in green ( scaled by a factor 100). Both are plotted as
functions of $i_{trig}$. 

\begin{figure}[h]
\centering
\includegraphics[scale=.3]{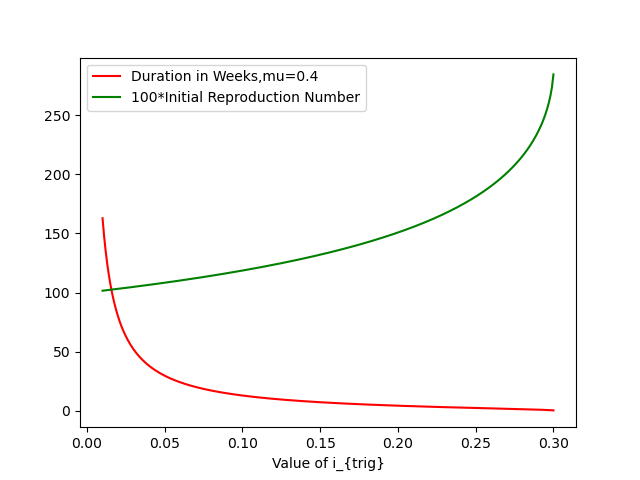}
\caption{Duration and intensity of mitigation as functions of
  $i_{trig}$ when $R_0=3$.}
\label{timeVsIma}
\end{figure}

To minimise this long constraint when $i_{trig}$ is small, we consider a
variant of this scenario. For this variant, when the ratio $i$ of infected
people reaches $i=i_{trig}$, a mitigation
is set up as above to preserve the health system. But the mitigation
is relaxed sooner in comparison to the previous scenario :
as soon as it is possible to stop the mitigation without overwhelming
the health system in the
future, the mitigation is relaxed.  For instance, suppose that the
health system is totally full when $i=i_{hosp}=0.15$. A mitigation is
triggered when $i=i_{trig}=0.05$, and it is maintained for a moment so
that  $i(t)$ stays blocked at the constant value $i=0.05$. When the mitigation is relaxed,
the ratio $i$ of infected people increases again from $i=0.05$, but
it never exceeds $i_{hosp}=0.15$. 
In other words, a 
rebound of the epidemic occurs
but the health system remains not full and viable after the mitigation
is over. This strategy   $S_{i_{trig},i_{hosp}}$ is illustrated in
figure \ref{avecRebond}. It depends on the two constants $i_{trig}$ and
$i_{hosp}$. The constant $i_{hosp}$ depends on the health system of the
country and is not changeable in the short term. In contrast,
different values of $i_{trig}$ are possible and  lead to different
public policies. 

\begin{figure}[h]
\centering
\includegraphics[scale=.3]{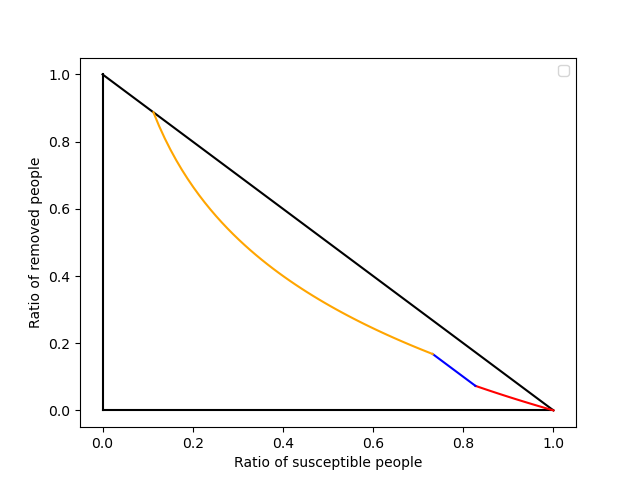}
\caption{A mitigation is settled, then relaxed with a possible rebound.}
\label{avecRebond}
\end{figure}

The duration of the mitigation in the scenario
$S_{i_{trig},i_{hosp}}$ is shown as a
function of  $i_{trig}$ in Figure \ref{saturationDuration}.
\begin{figure}[h]
\centering
\includegraphics[scale=.3]{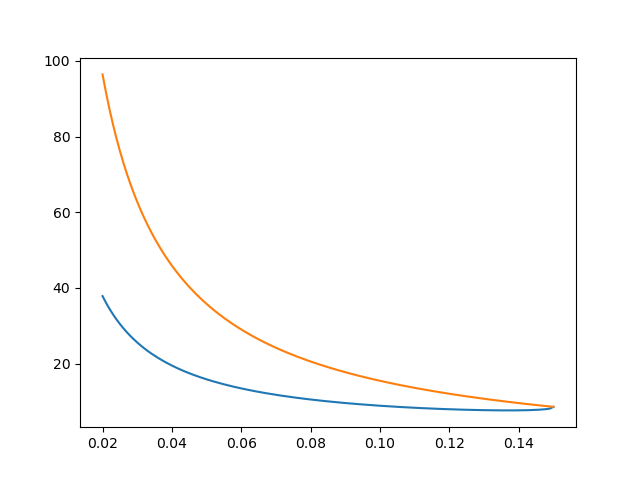}
\caption{Duration of the mitigation in weeks as a function of
  $i_{trig}$ for $i_{hosp}=0.15$, $R_0=3$.}
\label{saturationDuration}
\end{figure}
We see that the duration in weeks (represented by the blue curve) is
shorter when a rebound is
allowed, in comparison to the previous scenario without rebounds (
orange curve). The difference is significant, thus the objective of
lowering the time of mitigation by allowing a rebound is achieved. 

Besides a shorter mitigation time, there are several
differences that make the scenario with rebound quite different
from the scenario without rebound.

\begin{figure}[h]
\centering
\includegraphics[scale=.3]{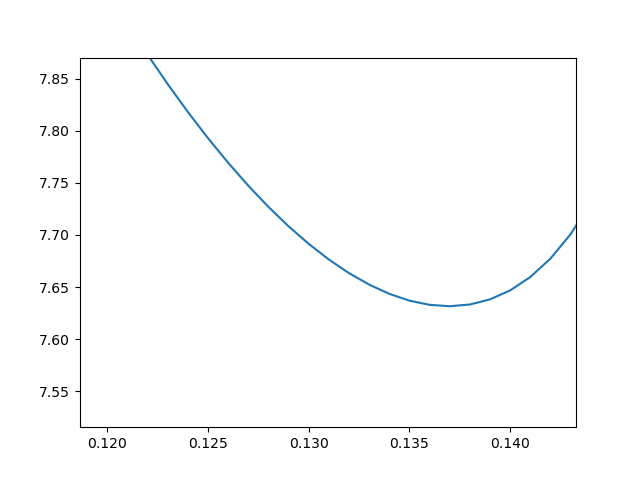}
\caption{Duration of the mitigation in weeks as a function of
  $i_{trig}$ for $i_{hosp}=0.15$, $R_0=3$.}
\label{zoomOnDuration}
\end{figure}
First, when a final rebound is allowed,
the duration is not a decreasing
function of $i_{trig}$ any more, as illustrated by a zoom on the
previous blue curve ( Figure \ref{zoomOnDuration}).
The minimal duration of around $7.6$ weeks for the mitigation
is obtained for $i_{trig}=i_{min}:=0.136$.

Second, the variation of $r_\infty$ is different too. In the scenario
without rebound, an early intervention lowered the ratio
$r_{\infty}$ at the price of a longer and harsher
mitigation. In the scenario with rebound, a harsher
or longer mitigation is not rewarded by a smaller $r_{\infty}$.
All strategies $S_{i_{trig},i_{hosp}}$ have the same
ratio $r_{\infty}$ of finally infected people independently of $i_{trig}$ for a fixed hospital capacity
$i_{hosp}$. In other words, the level $i_{trig}$
that triggers mitigations does not influence how many people will
be finally ill or dead (Theorem \ref{thm:withRebound} in the appendix).
This surprising phenomenon is illustrated in figure \ref{2rebounds} with $i_{hosp}=0.15$ and
two different values of $i_{trig}$. 
\begin{figure}[h]
\centering
\includegraphics[scale=.3]{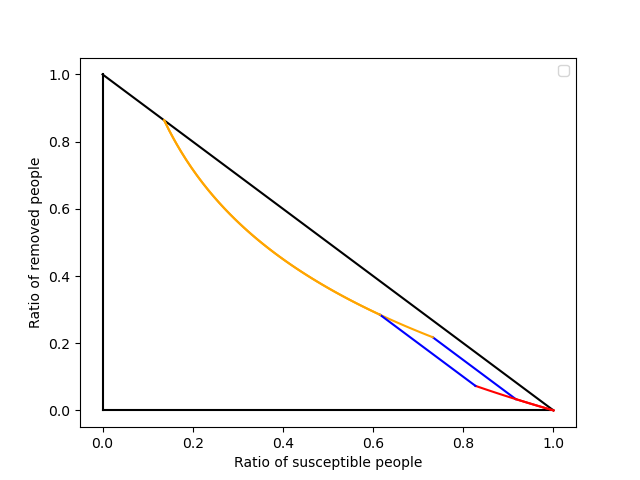}
\caption{Mitigations with   $i_{hosp}=0.15$, $i_{trig}=0.05$ or $0.10$. }
\label{2rebounds}
\end{figure}
The mitigations in blue are triggered when $i_{trig}=0.05$ and $i_{trig}=0.10$ respectively.
They are relaxed as soon as $i_{hosp}$ is never exceeded. The value
of $r_{\infty}$ which is common to  the two mitigations
is the $r$-coordinate of the point at the end of the yellow curve. 

As a consequence, the mitigations with $i_{trig}<i_{min}$
must be rejected. Indeed,  they are longer and harsher
than the mitigation with $i_{trig}=i_{min}$ and the supplementary
constraint is not compensated by an amelioration of $r_{\infty}$. 

We have excluded the cases $i_{trig}\in ]0,i_{min}[$. 
Let us now consider the remaining range $i_{trig}\geq i_{min}$.
In this range, all values of $i_{trig}$ may be considered
and lead to non comparable mitigations. More precisely, 
if mitigations are triggered above the minimal load
$i_{min}$, then
a higher load $i_{trig}$ gives a longer but less intensive mitigation. 
Thus a political trade-off between intensity and duration of the
mitigation is required to make the choice. Some people may prefer
a short intensive mitigation ($i_{trig}$ close to $i_{min}$) while other
people may prefer a long cool mitigation ($i_{trig}$ close to $i_{hosp}$).

The remarks formulated on the example are  illustrations of general
qualitative results proved in Theorem \ref{thm:withRebound}. The theorem  can be summarised
as follows.  In the variant with a rebound allowed,
the choice of the level $i_{trig}$
which triggers the mitigation has no impact on the
share $r_{\infty}$ of finally people infected. The choice of $i_{trig}$ impacts only
the subjective human or economical cost of the mitigation, but the
direct burden of the disease remains unchanged.
There exists a minimal load $i_{min}$  characterised by the following
properties. If $i_{trig}<i_{min}$, the strategy is to be rejected
as the mitigation is unnecessarily long and harsh for the same result. 
All the choices with $i_{trig}\in [i_{min},i_{hosp}]$ are possible and correspond to
different trade-offs between length and intensity : $i_{trig}$ closer
to $i_{min}$ corresponds to a shorter and harsher mitigation. In the example with
$i_{hosp}=0.15$, we have $i_{min}=0.136$. Since $i_{min}$ is close to
$i_{hosp}$, this means that the
mitigation must be triggered when the health system is nearly full.
Other simulations also give $i_{min}$ close to $i_{hosp}$. These simulations  express the
idea that, for the scenario with rebound within a sir-system, it is
a wrong idea to  anticipate much 
and to launch mitigation measures far before the saturation
of the health system. 

% We formalize  the order on the constraint
% when the coefficient $\beta(t)$ is varying with time.  We need to 
% compare two strategies with constraints $\beta_1(t)$ and
% $\beta_2(t)$. We give two possibilities : one using a cost function,
% the other using diffeomorphisms on the trajectories.

% The cost function is specific to each person. Each person has its own
% appreciation for instance of the difficulties arising from school
% closures. Thus, each person will have its own appreciation of the cost
% of a strategy. We can say that a strategy $S_1$ is preferable at a
% scientific level than $S_2$ if for every individual cost function,
% the total cost $c(S_1)$  of $S_1$ is smaller than $c(S_2)$.
% Otherwise, this a political trade-off : some people may prefer $S_1$
% whereas some other may prefer $S_2$.

% Of course it is not possible in practice to check an infinite number
% of cost functions. Our second approach using diffeormorphisms
% is introduced to bypass this difficulty : the existence of a suitable
% diffeomorphism between the trajectories of $S_1$ and $S_2$
% imply that $S_1$ is preferable than $S_2$. Thus it is a tool to check
% concretly in practice a concept not clearly checkable a priori. 

\subsection*{Temporary versus definitive mitigations, and temporal shifts versus
  improvements} 

The idea ``the sooner the restrictions, the better'' is often implicit or explicit
in the debate. For instance in  \cite{huffPostTotOuTard}, several
epidemiologists called for early mitigation measures. 
It is not supported by the above simulations and the
scenarios we studied.
As this may be surprising for many readers, we precise
in this section where the misunderstandings come from
and the underlying  phenomenons that explain this apparent paradox.

In this article, we consider \textit{temporary }
mitigation policies : the time of intervention is finite and then
people return to their normal life. The duration of the
intervention may be long, but it is finite. Considering
instead infinite time intervention can alter the assessment 
of a strategy.
For instance,  suppose that 
a starting  epidemic is annihilated with  
a drastic mitigation launched at the very beginning when the first
people are infected ; then the epidemic
never starts again if the mitigations go on forever and if normal
life never returns.
However, for finite time
interventions, the mitigation measures
eventually stop.  When normal life starts again, the situation is similar to the situation before
the mitigation measures, with a na\"\i ve population and no
immunisation. The epidemic 
will rise again  from the few remaining viruses or from the viruses
imported from abroad ( see for instance the simulations
by  Ferguson et al. \cite[fig.3]{ferguson}  where infections
rise after the mitigations are relaxed).
The problem has been postponed, rather than solved, by the finite time
mitigation. 

We consider the mortality in the long run
whereas other papers in the literature consider the
mortality at a precise date \cite{sofonea}. This may lead to
conclusions which are apparently in contradiction, but the 
results of the two approaches turn out to be compatible once the paradox is understood.
Suppose that a first strategy with a rebound of cases
after relaxation is
settled more early than a second more efficient strategy.  The first 
strategy appears to be preferable if the evaluation occurs before the
rebound of cases or 
when the second strategy has not yet been launched. But if the burden
of the epidemic is looked at later, when the
efficient late strategy has produced its effects, then the conclusion
becomes opposite. This phenomenon explains why  Sofonea et al.  \cite{sofonea}
find good estimates for  early mitigations 
whereas our conclusions are inverse for the long term.  
As an illustration, we revisit the example of Figure
\ref{2mitigations}. We  draw the epidemic
with the ratio of removed people $r(t)$ as a function of
$t$. After 26 weeks, the red mitigation seems preferable, but in the
long run, the opposite conclusion holds.    
\begin{figure}[h]
\centering
\includegraphics[scale=.40]{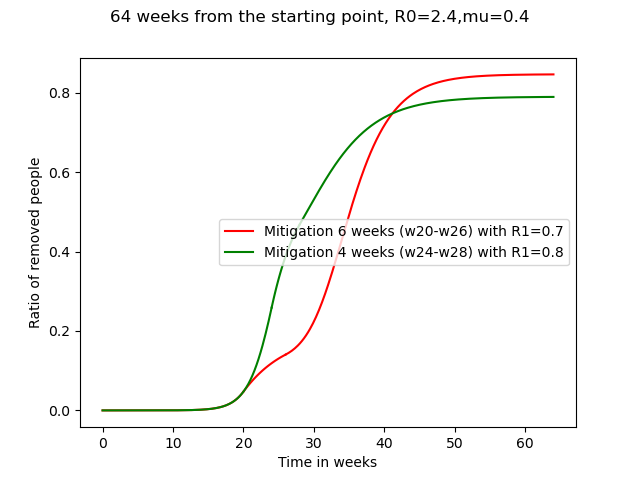}
\caption{Two mitigations with different plannings.}
\label{2mitigations_temporal}
\end{figure}

If only finite time interventions are allowed, and if assessment is
done in the long term, shifting the epidemic with no improvement of the
situation is not neutral, it is a waste of resources.  If some
measures are politically sustainable for 3 months, and if one month
is spent in a inefficient set of measures which postpones the problem, then only two months of mitigation
policies are left to improve the situation. This explains why in
many simulations, early interventions are inefficient in the context of finite
time interventions : They are temporal shifts rather than
improvements, time is wasted.

\subsection*{Energy $h$ of the system} 

To make more rigorous the distinction 
between temporal shift and improvements induced by an intervention, we introduce
the
energy $h$ of the system. A mitigation that lowers $h$  ameliorates the situation
while a mitigation that lets $h$
roughly unchanged acts as a temporal shift . In formula, the energy of a point $(s,r)$ is
% As an illustration, an absolute lockdown that would
% completly block the transmission of the epidemic is
% shown  in figure ... :  from $P_0$ to $P_1$, no one is infected
% ($s$ remains constant). Obviously, such a lockdown does not
% exist in real life, but it is interesting to keep it in mind as a
% limit case.
\begin{displaymath}
  h(s,r)=R_0-1-\ln(R_0s)-R_0r.
\end{displaymath}
The function $h$ is positive for every possible $(s,r)$ and minimum at point
$P_{herd}=(1/R_0,1-1/R_0)=(s_{herd},r_{herd})$, where its value is $0$. 
We denote by $C_h$ the equienergy curve containing the points
$(s,r)$ with energy level $h$, \textit{i.e.} $h(s,r)=h$.
They are the red curves
on figure \ref{levelLines}. In the absence of mitigation, the energy of the point $M(t)=(s(t),r(t))$ is
unchanged, \textit{i.e.} the world $M(t)$ moves along these red curves.
The herd point
$(s_{herd},r_{herd})$ is drawn in green on the figure. 
\begin{figure}[h]
\centering
\includegraphics[scale=.50]{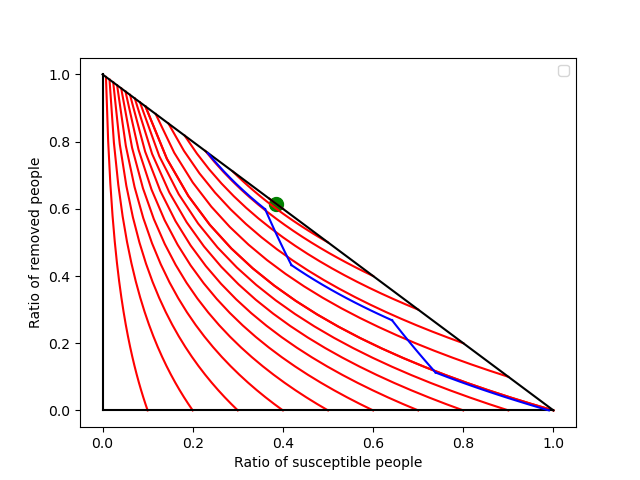}
\caption{Equienergy curves and a pandemic with 2 mitigations}
\label{levelLines}
\end{figure}
The blue curve on the figure is the trajectory of an epidemic 
where two mitigations have been launched.
%During the periods of normal life, we see that the epidemic progresses along the
%equienergy curves. 
During the 2 mitigation periods, the epidemic crosses these level
lines, dissipates energy, and the energy $h$ in the final situation
is smaller than the initial energy.   Thus $h$ is analogous to energy in physics.
It is constant when the system evolves freely (no mitigation),
and it diminishes when breaking/mitigation dissipates energy from the system.

The geometry of the equienergy curves show that
$h$ is an indirect measurement of the total burden ( past and
future, \textit{i.e.} including the
mortality to come) in the absence of further intervention. Two points
on the same equienergy curve go to the same point
at infinity if no mitigation measures are set any more :
two situations
 $S_1=S(s_1,i_1,r_1)$ and  $S_2=S(s_2,i_2,r_2)$ with the same value of
 $h$ will give the same number $r_{\infty}$ of finally infected people. 
 Moreover, in the absence of further mitigation,
 the higher the value of $h$, the more people will be  infected :  $r_{\infty}$ is an increasing function of $h$.
\textit{It follows that the goal of the policy maker is to propose measures
that  lower $h$ as much as possible before the mitigations are
definitively relaxed, with the minimal level of constraint on the
population}.

\subsection*{Variations of $h$ during mitigations }

How much $h$ decreases during a mitigation is governed by the
following differential equation.  Suppose
the epidemic is modelled by a sir system with
propagation number $R_0>1$ in the absence of mitigation, and
by a sir system with propagation number $R_1<R_0$
when some restrictions are applied.  During the mitigation,  $h$ is submitted to the differential equation
$$\frac{dh}{dt} =(R_1-R_0)i(t).$$ This equation
is of particular qualitative importance.
Indeed, the goal is to lower $h$ rapidly.
For a short time $dt$, the variation of $h$ is $dh=(R_1-R_0)i(t)dt$.
This means that for a fixed mitigation with number $R_1$ and
a fixed small duration $dt$, the decline of $h$ is more important when
$i(t)$ is large. A short time intervention is more
efficient when it is applied when the number of infected people $i(t)$
is large. This result is consistent and may be an explanation for the numerical observation of
\cite{ferguson} for an other model :''the majority of the effect of such a
[mitigation] strategy can be achieved by targeting interventions [...] around the peak of the
epidemic.''

At the other extreme,  if $i(t)=0$,
$dh=0$.  We recover the qualitative fact that mitigations applied 
when $i(t)$ is very small 
postpone  the problem with no improvement  
since $h$ does not decrease. In particular, an intensive mitigation
at the beginning of the epidemic is inefficient.

% To pursue our analogy,
% slowing down an epidemic when $i=0$ is equivalent to breaking
% a bike when it is stopped : the effort to break is wasted as
% the total energy of the system remains unchanged. 

\subsection*{On the minimal $r_\infty$}

Since the ratio of finally infected people $r_{\infty}$ is an
increasing function of $h$ after the last mitigation, and since $h$ is minimal at the herd
point, the inequality  $r_{\infty}> r_{herd}$ holds 
whatever the finite time mitigations. In other words, finite time mitigation measures cannot
be used to maintain the epidemic at level zero. In the long term,
the minimal share of people that have been infected  is at least
$r_{herd}$. Graphically, the limit of the red curves in figure
\ref{levelLines}
is always located above the herd point
$P_{herd}$.

This impossibility of a zero case strategy by
mitigations, or more generally of strategies to reach $r_\infty\leq
r_{herd}$, is valid only for the finite time strategies considered in
this paper. Figure \ref{finiteVsInfinite} compares a finite time and an
infinite time strategy. On the left part, an infinite time strategy is
settled and the limit point is  below the herd point,
\textit{i.e.} $r_{\infty}<r_{herd}$. On the middle part of the figure, the
same mitigation is relaxed after 20 weeks. There is a rebound when relaxing
occurs, and $r_{\infty}>r_{herd}$. On the right part, the
mitigation is relaxed after 60 weeks, which makes nearly no difference
with 20 weeks, apart from the delay in the 60 weeks case. The sporadic cases that remain
active yield quite the same rebound of the epidemic in both cases.

\begin{figure}[h]
\centering
\includegraphics[scale=.50]{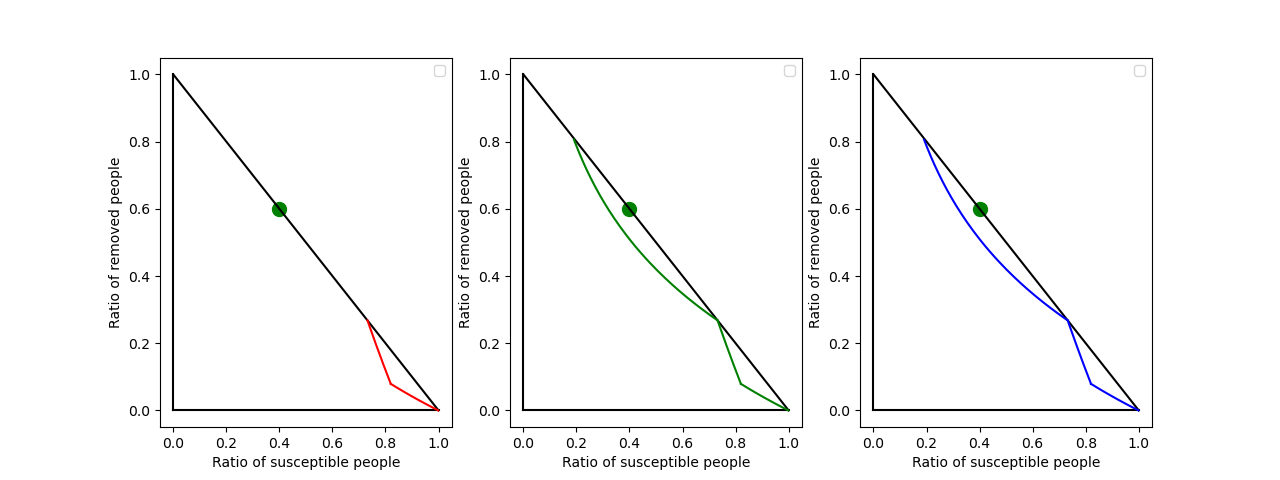}
\caption{Finite and infinite mitigations}
\label{finiteVsInfinite}
\end{figure}

This theoretical result is consistent with on the
ground situations for COVID-19. Several countries first tried to develop a zero
case strategy, and most of them finally desisted from
this strategy 
\cite{newzealand}.
Our analysis suggests a little more : For a given epidemic, it will be difficult if not
impossible to maintain $r<r_{herd}$ in the long run.

The inequality $r_{\infty}>r_{herd}$ for any finite time strategy
is an epidemiological analogue of the mechanical
situation with a bike on a sloppy road.  If an infinite time  breaking is possible,   the
bike can be stopped in the middle of the slope. In contrast, if the breaking time is finite,
the breaks are eventually released,  the bike will move to
the low point of the road and beyond. In this sense, the herd
line $s=s_{herd}$ is the analogue of the bottom of the slope. No
finite time breaking can stop the epidemic before it crosses the
herd line.

\subsection*{Dynamics and Inertia of the system after the herd ratio}

The analogy with the bike makes it easier to understand the
role of the herd immunity threshold, measured equivalently by
one of two herd ratios 
$s_{herd}=\frac{1}{R_0}$ or $r_{herd}=1-\frac{1}{R_0} $. In the vaccine pre-epidemic context, there are no
dynamics.  The herd ratio $r_{herd}$ is the proportion of people
that need to receive a vaccine to  nip in the bud any propagation
of the epidemic, preventing its launching from imported cases or
remanent cases in the population.
In a situation with dynamics, when the  epidemic has already started,
the situation is different. The epidemic will not stop instantly
when the herd ratio is reached. In this dynamical context, $r_{herd}$
and $s_{herd}$ still
make sense, but have a different interpretation : $i(t)$ will decrease
with time if $s(t)<s_{herd}$. In other words, $s_{herd}$ is the 
threshold that guarantees the fall of the number of infected people
with no mitigation measure.
Since  $ i(t)$ is proportional to the derivative $\frac{dr}{dt}$,
$i(t)$ is  a measure of speed
when one tries to minimise the total quantity $r_{\infty}$. 
The decrease of  $i(t)$  without mitigation after the rational
ratio $s_{herd}$ is the epidemiological counterpart to the fact
that a bike that reaches the
bottom of the slope will slow down without breaking. 

How far the epidemic will go beyond the herd line when 
mitigations are released is the analogue of the question of how far goes
a bike after the bottom of the slope. 
It depends on the energy $h$ of the system. If all the mitigations are
released at a point $(s,r)$ with energy $h(s,r)=h_0$, the share
$r=r_{\infty}$ of finally
infected people at infinity satisfies
the equation 
\begin{displaymath}
h(1-r_{\infty},r_{\infty})=h_0.
\end{displaymath}

Some comments in the media suggest that once the
herd ratio $s=s_{serd}$ is reached, the situation is under control and needs no more
supervision. Our models suggest quite the opposite !
Because of the inertia and of the dynamics, it may be necessary  to control and slow down the
epidemic after the
herd ratio is reached. Moreover, in many examples,
$i(t)$ is maximal at the herd ratio $s(t)=s_{herd}$.
This is true for instance when no mitigation is launched
before the herd ratio is reached. In such situations, 
the equation  $\frac{dh}{dt} =(R_1-R_0)i(t)$
tells that  the moment the herd threshold is reached  
co\"\i ncides with maximum effectiveness of the mitigation measures.
Thus, not only are the interventions often still necessary  when the 
herd ratio $s=s_{herd}$ is reached, but also, launching the mitigations measures around
the herd immunity ratio is good practice according to the sir-model.

\subsection*{ Building mitigations with small $r_{\infty}$}

What is the infimum possible $r_{\infty,opt}$ for the number $r_{\infty}$ of finally
infected people and what are the strategies to lead to this infimum ?
We have seen in the previous section that $r_{\infty}> r_{herd}$ for any finite time strategy.
In this section, we explain  that
$r_{\infty,opt}=r_{herd}$, which means that  the difference between $r_{\infty}$ and
$r_{herd}$ may be arbitrarily small for a well-designed strategy. 
We discuss the strategies that lead to
this infimum, \textit{i.e.}  the strategies with $r_{\infty}$ 
close to $r_{herd}$.

The na\"\i ve
and falsely convincing argument that the more restrictive mitigations give
the better results on $r_{\infty}$ is wrong. 
If the mitigation is too intensive, there will be a large rebound
in the epidemic. If the mitigation is too loose, the epidemic is not
stopped enough. An adequate calibration of the mitigation
is thus necessary to get an optimal result, not too intensive to avoid a
large rebound, nor too loose to sufficiently dampen the dynamics. 
This is illustrated in figure \ref{hardSmooth} with 3 mitigations
starting at the same point. The black curve is the epidemic when no
mitigation occurs. The
mitigation in blue is too loose and the trajectory goes beyond the
herd point. The mitigation in green is suitable towards the
herd point.  The mitigation in red is too
strict : a large rebound occurs after the mitigation is stopped. 
\begin{figure}[h]
\centering
\includegraphics[scale=.40]{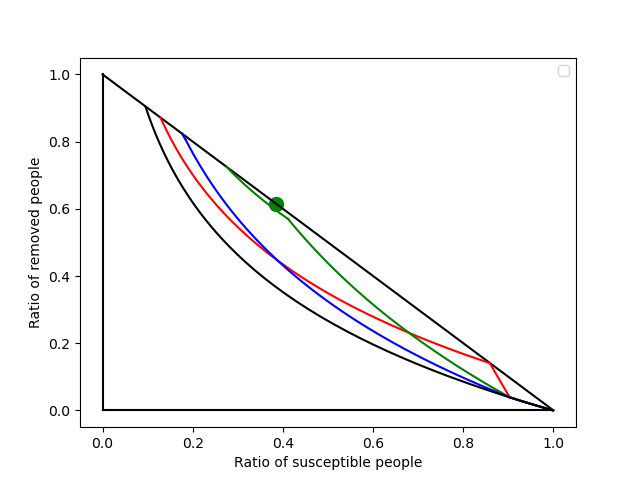}
\caption{three mitigations : too harsh, correct, and too loose}
\label{hardSmooth}
\end{figure}
The optimal constant mitigation is understood : its reproduction
number is 
\begin{align} \label{equationR1opt}
R_1=\ln(sR_0)/(1-\frac{1}{R_0} -r).
\end{align}
where we suppose that the world is in situation
$(s,r)$ with $r<r_{herd}$. 
Such a mitigation applied an infinite time would lead
to the herd point. In practice, the strategy is applied a finite long
enough time,
the herd point is approached, and $r_{\infty}$ is close to
$r_{herd}$. This is illustrated by the green curve in figure
\ref{hardSmooth} which uses the
above formula for $R_1$. 
% \begin{figure}[h]
% \centering
% \includegraphics[scale=.5]{oneShotToVaccinalPoint.png}
% \end{figure}

An other question is the timing.
Is it necessary to launch a strategy long before the
herd immunity threshold 
to settle a strategy with $r_{\infty}\simeq r_{herd}$ ? 
The answer is no from a theoretical point of view. 
As long as
$s(t)>s_{herd}$, it is not too to late to launch a strategy that leads to $r_{\infty}$ close
to $r_{herd}$.  For instance a mitigation strategy with $R_1$ as above works. If $s(t)<s_{herd}$, the
formula implies $R_1<0$ which is an impossible physically.
 In the formula,  $R_1$ tends to zero when
$s$ tends to $s_{herd}$, which is a rephrasing of the high intensity
required for a late mitigation.  This means that on the ground,
there are limitations besides theoretical limitations
since $R_1$ cannot be arbitrarily low.

The two questions, intensity and timing, are correlated. 
There is a large range of possibilities for the start
date of an optimal mitigation.  The formula for $R_1$ shows that
the later the start date, the
more vigorous the intervention must be. But on the other hand,
the later the intervention, the shorter
the duration of the mitigation for the same result. 
We can use this correlation between timing and intensity to revisit
and explain the initial example of figure \ref{2mitigations} or the
red mitigation in figure \ref{hardSmooth}.  
The early
mitigation of figure \ref{2mitigations} was inefficient  because the intensity was not consistent
with the starting point. The mitigation was  
too intensive  for its earliness, leading to poor results. 
% The following array illustrates 
% at the qualitative level different possible choices for an optimal
% intervention.
% $$
% \begin{array}{|c|cc|}
%   \hline
%   Start Date & Intensity & Length \\
%   \hline
%   Early      & Small     & Long \\
%   Late       & Large     & Short\\
% \hline                           
% \end{array}
% $$

To get order of magnitudes, we illustrate in table \ref{r1_time}
possible consistent values for the intensity and
duration of the mitigations.  We fix $R_0=3$ and $\mu= 0.33$.
An epidemic rises, and a mitigation is launched when
$s=s_{start}$ with the above optimal formula for $R_1$. The mitigation
is relaxed so that $r_\infty=l*r_{herd}$, with $l=1.1$ or $l=1.2$.
Thus by construction, all scenarios have a burden dependent on $l$ but not
on the choice of $s_{start}$ for a fixed $l$.
The table reports the values of $R_1$ and the duration
of the mitigation for different values of $s_{start}\in
]s_{herd}=\frac{1}{3},1]$.

  \begin{figure}[h]
    \centering
  \input{tableau4}
  \caption{R1 and duration of the mitigation}
      \label{r1_time}
   \end{figure}

Remark that there is some flexibility in the choice of $R_1$ for a
fixed $r_{\infty}$. When
$R_1$ is chosen small in the set of possible values, there will be a rebound in the epidemic whereas
a large $R_1$ will yield a strategy with no rebound after the
mitigation is relaxed. However, both strategies have the same $r_{\infty}$.
Thus the existence of a rebound is not an indicator of an overshoot
and of a badly calibrated mitigation.

\subsection*{Indeterminacy of $R_0$ and timing implications }
Many computations above rely on the estimate of the reproduction
number $R_0$.
For instance, the computation of $r_{herd}$, the energy $h$, the reproduction
number $R_1$ of the optimal
mitigation depend on $R_0$. However the value of $R_0$ is not
well known.
Determining $R_0$  is an important and difficult
question. See  \cite{r0_used_in_owid}, for the methodology
implemented in \textit{Our world in data}, and
the many references therein. 
In this section, we discuss the implications of this
indeterminacy on the choice of a strategy. In particular, we
exhibit a strategy implementable on the ground without
knowing $R_0$, thus bypassing the difficulty to estimate $R_0$.

There are many reasons that make it difficult to estimate
$R_0=\frac{\beta_0}{\mu}$. For a virus, $\beta_0$, hence $R_0$, depend on the variants that
propagate, they evolve with time. The number $R_0$ to be used in the modelling
can be different from the $R_0$ that has been computed
at the beginning of the epidemic because of the change in the
variants. This problem is amplified by the
interaction between vaccines and variants. 
If a vaccine has efficacy $e$
with $0<e<1$, and the share of vaccinated people is $r$, then
the fraction of the population protected by the vaccine is
$re$. These people are to be placed in the removed compartment
of the sir system. Thus some continuous exchanges
between the $s$-compartment and the $r$-compartment occur along with
the evolution of the variants and their interactions with vaccinated people, making the calibration of
$R_0$ difficult.

An other difficulty is that the sir model is valid only
locally because of heterogeneity. For a large territory,
one can choose a constant $R_0$ suitable to aggregate
the constants of each local area where homogeneity
makes sense. However, the local areas evolve independently
and the choice of $R_0$ may lead
to a coherent modelling only for a short time.
An other heterogeneity was pointed out in \cite{britton}.
The more the individuals have contacts, the more they spread the virus and the
sooner they are infected. Heavy spreaders are in average
removed sooner, and $R_0$ decreases with time. 
Using  the basic reproduction number $R_0$ computed
at the beginning of the epidemic is thus expected to yield
overestimation of $r_\infty$. 
% Could this unebe an explanation for the
% ending of the waves ? It is plausible that the $R_0$ suitable for the
% model after a few weeks of transmission is much lower than the initial
% $R_0$, leading to an end of the epidemic. Then new variants or new
% conditions are necessary later on that reactualise $R_0$ and
% lead to a new wave.

We may formalise these
remarks as follows. There is a propagation number $R_0(t)$ depending
on time, on the (unknown) distribution of
heterogeneity, on the immunity evolving with vaccination.
For a short enough period of time, all parameters can be seen as constants. The
time-varying $R_0(t)$ can thus be  approximated by a constant and the sir
model makes sense.

There exists
an optimal intervention which is as late and as strict as possible.
If  $s=s_{herd}=\frac{1}{R_0} $ in formula
\ref{equationR1opt}, we get $R_1=0$ \textit{independently of the value of $R_0$.}
This corresponds to a harsh mitigation. Moreover, in the absence of
mitigation, $s(t)=s_{herd}$ is 
reached at the moment the number of
infected people starts to decrease. Thus there exists an optimal
mitigation that starts when the ratio $i(t)$  of
people infected naturally decreases. In contrast to the estimation of
$R_0$ which is very difficult, the moment when  $i(t)$ starts to
decrease is easier to monitor. It can be estimated with a detection of the virus in the sewage
system. This has already done in France through the ``R\'eseau
Obepine'' \cite{obepine1}.

To sum up, the start time and the intensity of the mitigation should
be consistent, and this
consistency is difficult to achieve because of the indeterminacy of
$R_0(t)$ which is a constant
only locally in time. This suggests that the
following policy independent of $R_0$ could be considered.
The public authorities organise the
monitoring of the virus in the sewage system. 
When $i(t)$
decreases, it may be interpreted as  ``the herd immunity ratio is
reached''. Then it is a good moment to launch the communication to the
public to slow down the spreading
and the dynamics with a vigorous mitigation. 
We think that the feasibility of this proposal needs to be considered
as it has several advantages. The mitigation time is shorter than for
other intervention times, making the
acceptance of the strategy easier. The public confidence in the 
political decisions is preserved because there are fewer risks of
contradictory decisions induced by bad estimates of $R_0$.

\section{Supplementary}
\label{sec:supplementary2}

\subsection{Description of the model}

In this section, we describe the sir-controlled model that we use
throughout the article, which involves some non continuity considerations
that need to be clearly formulated. 

We consider an epidemic starting  at time $t=0$. Three functions
$s,i,r$ of the time $t$ are considered : $s(t)$ is the share of people  not infected up to $t$, 
$i(t)$ is the share of people infected at $t$ and $r(t)$ is the share
of people who were infected before $t$, but are cured at $t$.
By construction, for every
non negative $t$,  $s(t)+i(t)+r(t)=1$. We
do not consider births in the model.  
The classical sir equations are: 
\begin{displaymath}
\begin{cases}
  s'&=-\beta_0 s i\\
  i'&=-\mu i +\beta_0 s i\\
  r'&=\mu i 
\end{cases}
\end{displaymath}
where
\begin{itemize}
\item 
$\mu$ is a strictly positive constant and  depends on the
epidemiological context. It governs the speed at which infected
people are removed. 
\item $\beta_0$ is a  strictly positive constant and
  $\frac{\beta_0}{\mu}$ is the initial propagation number, \textit{i.e.} the average
  number of infections originated from the first infected persons. We use the
  classical notation $R_0=\frac{\beta_0}{\mu} $.
\end{itemize}

\begin{rem}
  The classical SIR-system is often presented using absolute numbers,
  whereas the above version considers ratios rather than sizes of
  populations. For
  instance, in \cite{martcheva}, the Kermack-McKendrick SIR epidemic
  model is presented with the number $I$ of infected people, and $N$
  the size of population, whereas we use the share $i=\frac{I}{N} $
  instead.  We use the lowercase notation {\bf sir-system}
  rather than the uppercase notation  $SIR$-system as a reminder of
  this choice. 
\end{rem}

When mitigation policies apply, $\beta_0$ is not a constant any more.
We consider a model where we replace the constant $\beta_0$ with
a function $\beta=\beta(t)$ depending on the time $t$. When no mitigation
strategy is set up, $\beta(t)=\beta_0$. When mitigation strategies
occur, $0\leq \beta(t)<\beta_0$. The condition
$\beta(t)>\beta_0$ is mathematically possible, but corresponds to
people gathering and 
transmitting the virus more than expected, thus is hardly realistic.  

On the other hand, $\mu$ is constant as before and is independent of
the mitigation strategy. 

Summing up, our model to include mitigation strategies is a derivation
of the classical sir model where $\beta_0$ is replaced by
a non negative function $\beta=\beta(t)$. In mathematical terms :
\begin{displaymath}
(*)\begin{cases}
  s'&=-\beta s i\\
  i'&=-\mu i +\beta s i\\
  r'&=\mu i 
\end{cases}
\end{displaymath}
In the following, we will use the expression "sir-controlled model"
for this model, or sometimes only "sir-model" for simplicity, assuming it is
implicit and clear that $\beta$ is not a constant. 

When a mitigation is launched at time $t$, a discontinuity occurs for
the function $\beta$. Thus we need to consider non
continuous functions $\beta(t)$ and we suppose only that $\beta(t)$ is
piecewise continuous on $[0,+\infty[$. In this non continuous context,
we define a solution of a controlled sir-system as follows. 
\begin{defi}
 We suppose that there exists a subdivision
$a_0=0<a_1< \dots <a_{k-1}<a_k=+\infty$ such that $\beta$ is continuous on
each $]a_j,a_{j+1}[$. Moreover, we suppose that for $j\in \{ 0,\dots,k-1\}$, the right limit
$\beta(a_j^+):=\lim_{t \to a_j,t>a_j}\beta(t)$ and for $j\in \{
0,\dots,k-2\}$, the left limit
$\beta(a_{j+1}^-):=\lim_{t \to a_{j+1},t<a_{j+1}}\beta(t)$ exist. 
In general,  $\beta(a_{j+1}^-)\neq \beta(a_{j+1}^+)$. \\
Let $\tilde \beta_{j,j+1}:[a_j,a_{j+1}]\to \RR^+$ be the continuous function that extends
$\beta:]a_j,a_{j+1}[\to \RR^+$. A solution of the sir
controlled system is a triple of functions  $(s,i,r)$ defined on
$[0,+\infty[$ satisfying
\begin{displaymath}
(**)\begin{cases}
  s'&=-\tilde \beta_{j,j+1} s i\\
  i'&=-\mu i +\tilde \beta_{j,j+1} s i\\
  r'&=\mu i 
\end{cases}
\end{displaymath}
on each bounded interval $[a_j,a_{j+1}]$ for $j\leq k-2$ and on
the last unbounded interval  $[a_{k-1},+\infty[$.
In particular, $s$
has a left and right derivative at each point $a_i$, $i\in
\{1,\dots,k-1\}$ and a derivative at $a_0$. 
\\
We define $R(t)=\frac{\beta(t)}{\mu}$. If $\beta(t)=\beta_0$ for all
$t$, then $R(t)=R_0$ is the classical propagation  number.
\end{defi}

\subsection{Existence of solutions}
\label{sec:general-non-sense}

In this section, we prove that under our hypothesis for $\beta$, the sir
controlled systems have a unique solution,  with regularity
and monotony properties. These are standard facts when $\beta$ is a constant( see for instance
\cite{martcheva})  and the proof extends easily when
$\beta$ is a piecewise continuous function. 

\begin{thm} \label{thm:descriptionOfTheSolutions}
  Let $(s_0,i_0,r_0)$ with $s_0>0,i_0>0,r_0\geq 0$ and
  $s_0+i_0+r_0=1$. Then there exists a unique
  solution $(s,i,r)$ of the sir-controlled system $(**)$ defined on
  $[0,+\infty[$ satisfying $(s(0),i(0),r(0))=(s_0,i_0,r_0)$.
  Moreover, the solution satisfies the following properties:
  \begin{itemize}
  \item $s,i,r$ are continuous and their restrictions on the intervals
    $[a_j,a_{j+1}]$, $j\in \{0,\dots,k-2\}$ and $[a_{k-1},+\infty[$ are
    $C^1$, 
  \item $\forall t\geq 0$, $s(t)>0$ and $i(t)> 0$, 
  \item $\forall t>0$, $r(t)>0$,
  \item $r$ is strictly increasing and $s$ is decreasing,
  \item $\forall t \geq 0$, $s(t)+i(t)+r(t)=1$, 
  \item The limits $s_\infty:=\lim_{t\to \infty} s(t)$,
    $i_{\infty}:=\lim_{t \to \infty}i(t)$ and 
    $r_{\infty}=\lim_{t\to \infty}r(t)$ exist. 
  \item $i_{\infty}=0$, $s_{\infty}+r_{\infty}=1$.
  \end{itemize}
\end{thm}
\begin{proof}
  We first consider the equations on the interval $[a_0=0,a_1]$. Let
  $I\subset [a_0,a_1]$ be the maximal interval where the solution with
  initial condition   $(s(0),i(0),r(0))=(s_0,i_0,r_0)$ is defined. \\
  If $s(t)=0$ for some $t\in I$, then $s(t)=0$ for all $t \in I$ since
  $s$ satisfies a first order linear equation  $s'=(-\beta
  i)s$, contradicting the value of $s(0)$. Similarly, $i(t)$ cannot
  vanish. Thus $i(t)>0$ and $s(t)>0$. It follows by derivation
  that $r$ is strictly
  increasing and $r(t)>0$ for $t>0$, and that $s$ is decreasing. By the sir-system, $s+i+r$ is a
  constant function since its derivative is zero, and its initial
  value is 1.
  \\
  Let $b=sup(I)$, \textit{i.e.} $I=[a_0,b[$ or $I=[a_0,b]$ with $b \leq
  a_1$. The limits $s(b^-),i(b^-),r(b^-)$ at $b$ exist since $s$ is
  positive decreasing, $r$ is increasing bounded by $1$, and
  $s+i+r=1$.
  \\
  If $b=+\infty$, it remains to prove that $i_{\infty}=0$. If
  $i_{\infty}>0$, the equation $r'=\mu i$ shows that
  $r_{\infty}=+\infty$, contradiction. Thus $i_{\infty}=0$ and
  $r_{\infty}+s_{\infty}=1$ follows. 
  \\
  If $b<+\infty$,  the limits of the derivatives $s'(b^-),i'(b^-),r'(b^-)$
  exist by the sir-system and the limits of $s,i,r$.  Therefore, 
  the solution can be defined at $b$. By maximality of $I$, it
  follows that $b=a_1$. We conclude by induction, replacing $t=a_0$ by
  $t=a_1$, with fewer intervals for the definition of $\beta$. 
\end{proof}

\begin{prop}\label{prop:i_increasing_then_decreasing}
  Let a sir-system with $\beta(t)=\beta_0$ for $t\geq t_0$. If $s(t_0)<\frac{1}{R_0}
  $, then $i(t)$ is strictly  decreasing for $t\geq t_0$. If
  $s(t_0)>\frac{1}{R_0}$, then there exists a $t_1>t_0$ such
  that $i$ is strictly increasing on $[t_0,t_1]$ and strictly decreasing on
  $[t_1,+\infty[$. Moreover $t_1$ is characterised by
  $s(t_1)=\frac{1}{R_0} $. 
\end{prop}
    \begin{proof}

    If $s(t_0)<\frac{1}{R_0}$, then  the function $\beta_0 s -\mu$ is negative at $t_0$. As $s$ is decreasing, the same holds for $t\geq t_0$. By the sir-system, $i'=  
  ( \beta_0 s -\mu)i$ and then $i'(t)<0$ for $t\geq t_0$.  
  
  If $s(t_0)>\frac{1}{R_0}$, then by the same equation of the
  sir-system, $i$ is at first increasing. Since $i_{\infty}=0$, the
  function $i$ must have a maximum at some
  $t_1>t_0$.  At this point  $i'(t_1)=\beta_0 s(t_1) -\mu=0$,
  \textit{i.e.} $s(t_1)= \frac{1}{R_0}$. Since $s$ is decreasing,
$i$ is at first increasing, reaches its maximum when
$s(t_1)=\frac{1}{R_0}$ and then is  decreasing.
The increase and the decrease of  $i$ are strict, otherwise
$i$, $s$, and $r=1-i-s$ are constant on some interval,
whereas $r$ is  strictly increasing.
\end{proof}

\begin{coro}
  If a mitigation is stopped at time $t_1$ with $s(t_1)>s_{herd}$,
  then $i$ will be increasing for $t>t_1$ during a certain time. In
  other words, a new wave is rising. 
\end{coro}

\subsection{Foliation of the triangle and solutions of the sir-system}
\label{sec:geometry-triangle}

Our approach to study the solutions of the sir-system is to work in
the plane $\RR^2$ (rather than $\RR^3$)
with coordinates $(s,r)$,  forgetting the quantity $i=1-s-r$.
When $\beta(t)$ is a constant function, the trajectories of the solutions
in $\RR^2$ are included in a set with equation $H(s,r)=c$ for some
function $H$ and some constant $c$. The coordinates $(s,r)$ lies in a
triangle $T$ and 
the loci $H(s,r)=c$ when $c$ varies draw a foliation on $T$. In this
section, we introduce the foliation and we study its geometry (
Theorem \ref{thm:foliation}).  We then derive the
consequences of this geometrical study 
on the solutions of the sir-system associated to $\beta$
(Proposition
\ref{prop:trajectoriesInFoliations}). 

\begin{defi}
  Let $T\subset \RR^2$ be the set of points $(s,r)$ with $s>0,r\geq
  0$ and $s+r\leq 1$. Let $R\geq 0$ and $H:]0,+\infty[\times \RR  \fd \RR$ with
  $H(s,r)=\ln(s)+Rr$. If $R\geq 1$, we let $s_{herd}=\frac{1}{R} $,
  $r_{herd}=1-\frac{1}{R} $, $i_{herd}=0$, $M_{herd}
  =(s_{herd},r_{herd})$.  A $R$-leaf $C_{c}$ is a set in
  $T$ defined by $C_c:=\{M \in T, s.t. \ H(M)=c\}$ for some constant
  $c$. We let $i$ be the function on $T$ defined by $i(M)=1-s(M)-r(M)$
  where $s,r$ are the coordinate functions. 
\end{defi}

\begin{figure}[h]
\centering
\includegraphics[scale=.40]{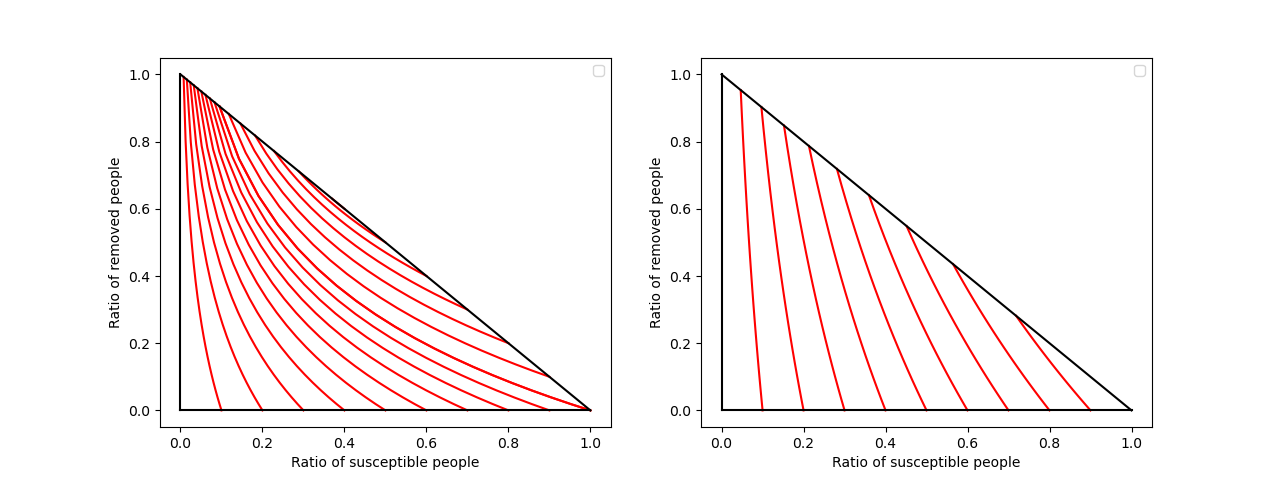}
  \caption{Leaves for $R=2.6$ and $R=0.8$}
      \label{leaves}
\end{figure}

The leaves for $R=2.6$ and $R=0.8$ are shown in figure \ref{leaves}. Later on, we
will use this foliation for $R=R_0$ the basic reproduction number
but also for $R=R_1$, the reproduction number induced by some
mitigation.

\begin{thm}\label{thm:foliation}
  \begin{enumerate}
  \item The $R$-leaves $C_c$ form a foliation of $T$, which means for
    us that  the
    leaves $C_c$ are smooth and connected. Moreover, the leaves are
    compact. 
  \item If $R\geq 1$, the point $M_{herd}$ is a leaf reduced to a
    point. The other leaves are curves.
  \item If $R\leq 1$, $M=(1,0)$ is a leaf reduced to a
    point. The other leaves are curves.
  \item For any leaf $C_c$, there exists a unique point $M_\infty$ on $C_c$
    such that $r(M_\infty)=\max\{r(M), \ M\in C_c\}$. Moreover,
    $s(M_\infty)=\min \{s(M), \ M\in C_c\}$ and $i(M_\infty)=0$
    ( When precision is required, we denote
    $M_{\infty}=M_{\infty}(C_c)$).
\item For any leaf $C_c$ with $c\geq 0$,there exists a unique point
  $M_{init}$  on  $C_c$ ( sometimes denoted $M_{init}(C_c)$ )
    such that $r(M_{init})=\min\{r(M), \ M\in C_c\}$. Moreover,
    $s(M_{init})=\max \{s(M), \ M\in C_c\}$ and $i(M_{init})=0$.
  \item If a leaf $C_c$ is non empty for $c> 0$, then $R> 1$. 
  \item A point $M\in T$ is a point $M_\infty$ of some leaf $C_c$ if
    and only if $i(M)=0$ and $s(M)\leq min(1,\frac{1}{R}) $. 
  \item Let $C_c$ be a leaf. If $C_c$ is reduced to a point, then
    $T\setminus C_c$
    is connected , otherwise
    $T\setminus C_c$ has 2  components.  
  \end{enumerate}
\end{thm}
\begin{proof}
  The function $H(s,r)$ increases  strictly when $s$ increases. Hence
  in order to study the maximum  of $H$, we can restrain it to the diagonal 
 $\{(s,r) , s> 0, r\geq 0, s+r=1\}$, or equivalently we  study the function $H(s,1-s)$ on the interval $]0,1]$. The derivative $H(s,1-s)'=\frac{1}{s}-R$. 
 shows that the maximum of $H$ is obtained for $s=\frac{1}{R}$ if
 $R\geq 1$ and for $s=1$ is $R<1$. It follows that the maximum of $H$ is
 realised on a unique point of $T$ which is a leaf reduced to a
 point $M_{max}$. If $R\geq 1$, $M_{max}=M_{herd}$, otherwise,
 $M_{max}=(1,0)$. % Moreover, in the case $R<1$, the non empty leaves $C_c$
 %satisfy $c\leq H(1,0)=0$. 

Let $M_1=(s_1,r_1)$ and $M_2=(s_2,r_2)$ be two points on the same leaf $C_c$ 
with $r_2<r_1$. The graph $\Gamma$ of the
function$[r_2,r_1]\to \RR$, $r \to s=e^{c-Rr}$ is  included in the
locus $H=c$ and contains $M_1$ and $M_2$. To prove the connectedness of $C_c$, it remains to prove
that $\Gamma\subset T$.  By convexity of the exponential function, $\Gamma$ lies below the segment
 $[M_1,M_2]$ and by monotony, $\Gamma$ is included in the rectangle
 $[s_1,s_2]\times [r_2,r_1]$. The part of the rectangle below
 $[M_1,M_2]$ is the triangle $Conv(M_1,M_2,M_3)$ with
 $M_3=(s_1,r_2)$. By convexity of $T$, it follows that $\Gamma \subset
 Conv(M_1,M_2,M_3)\subset T$.

 A point $(s,r) \in C_c$ satisfies $r\leq 1$ hence $s\geq
 \epsilon:=e^{c-R}$. Thus $C_c$ is compact as the intersection
 of the closed set $\ln(s)+Rr=c$ with the compact $T\cap \{s\geq
 \epsilon\}$. It follows by compacity and connectedness that
 $r(C_c)=[r_1,r_2]$ for some $r_1,r_2\in [0,1]$ and that $C_c$ is the
 graph $\Gamma$ of $r\in [r_1,r_2] \to s=e^{c-rR}$. The smoothness of
 $C_c$ follows.

 The existence and uniqueness of $M_{\infty}$ and $M_{init}$ in items $4$ and $5$
 follow from the description of $C_c$ as a graph of a monotonous
 function of $r$ in the previous paragraph. The condition $i(M_\infty)=0$
 holds : otherwise, $M_{\infty}$ is in the interior of the triangle
 $T$ and there would exist $r_3>r_2=r(M_\infty)$ with
 $(e^{c-Rr_3},r_3)$ in $T$, contradicting the
 maximality of $r(M_\infty)$ on the leaf. 
 Similarly, $M_{init}$ exists by compacity and
 satisfies the maximality conditions by monotony.
 If $c=0$, then obviously $M_{init}=(1,0)$ and point 5) is true.
 If $c>0$, 
 we have
 $r(M_{init})>0$. If $i(M_{init})\neq 0$, then
 $M_{init}$ is in the interior of the triangle $T$ and we contradict
 the minimality of $r$ as above. 

If $R\geq 1$, $\ln(s)+Rr\leq \ln(s)+R(1-s)\leq 0$ for $s\leq 1$. This
proves the sixth item. 
 
 By construction, a point $M_{\infty}$ is a point on a $R$-leaf with
 $r(M)$ maximum and we know that $i(M_{\infty})=0$.  Thus
 a point $M\in T$ is a point $M_\infty(C_c)$ for some $c$ if and only if $i(M)=0$
 and $r(M)=\max \{r(M')\}$ where $M'$ runs through the points in $T$
 satisfying $i(M')=0$ and $H(M')=c$. Rephrasing informally, the points
 $M_\infty$ are the leftmost points of $T$ on the segment $i=0$ among those
 $M'$ which have the same $H$-value.

 For item 7), if $R<1$, we need to show that any
 point $M$ with $i(M)=0$ is equal to $M_\infty(C_{c})$ for some
 $c$. We let $c=H(M)$.   The function $s \to
 H(s,1-s)=\ln(s)+R(1-s)$ is strictly increasing, thus $M$ is the 
 unique point $M$ with $i(M)=0$ and  $H(M)=c$. Therefore
 $M=M_{\infty}(C_c)$.

 If $R\geq 1$,  the function $s \to
 H(s,1-s)=\ln(s)+R(1-s)$ is strictly increasing from $s=0$ to
 $s=\frac{1}{R} $, then strictly decreasing from $s=\frac{1}{R} $ to
 $s=1$. If $c:=H(M)<H(1,0)=0$ or if $c=H(\frac{1}{R},1-\frac{1}{R} )=:H_{max}$, then  as above $M$ is the 
 unique point with $i(M)=0$ and  $H(M)=c$. Therefore
 $M=M_{\infty}(C_c)$. If $c\geq 0$ and $c<H_{max}$, then there are two
 values $s_0,s_1$ solutions of $H(s,1-s)=c$. The smallest value $s_0$
 satisfies $s_0<\frac{1}{R} $ while $s_1>\frac{1}{R} $. Thus
 $M_{\infty}(C_{c})=(s_0,1-s_0)$ and $M_{init}(C_c)=(s_1,1-s_1)$. It follows that the points
 $M_{\infty}$ are the points $(s,1-\frac{1}{s} )$ with $s\leq
 \frac{1}{R} $.

We have exhibited some leaves that  are reduced to points. We show now
that the other leaves are curves. We know that a leaf $C_c$ is a graph of a function
$r\in [r_1,r_2]\to e^{c-Rr}=s$ parameterised by the coordinate $r$. In particular, $C_c$ is
reduced to a point if $r_1=r_2$ and $C_c$ is a curve otherwise. In particular,
if $C_c$ contains at least two different points, it is a curve.
\begin{itemize}
\item If $R<1$, then $c\leq 0$, and $C_c$ contains the points
  $M_\infty(C_c)$ and $(e^c,0)$. If $c<0$, the 2 points are different
  and $C_c$ is a curve. If $c=0$, we have already seen that $C_0$ is a
  point.  
\item If $R\geq 1$ and $c<0$, we proceed as above. If $R \geq 1$,
  and $c\in [0,H_{max}[$, then $C_c$ contains the two distinct points
  $M_{\infty}(C_c)$ and $M_{init}(C_c)$ and it is a curve. 
\item In the last
  case $R\geq 1$ and $c=H_{max}$, we have already proved that
  $C_c=\{M_{herd}\}$ is a point.
\end{itemize}

For the last item, each leaf $C_c$ is connected, thus the connectedness problem reduces
to connect the various leaves $C_c$, or equivalently the various
points $M_{\infty}(C_c)$. The set of points $M_{\infty}(C_c)$ is the
set  $\{(s,1-s),\ with\ s\in I:=]0,\min(1,\frac{1}{R} )]\}$. Thus it
is connected since $I$ is connected. 

When
a curve $C_c$ is removed from $T$,  $I$ is replaced with $J=I\setminus
\{s(M_\infty(C_c))\}$. If $C_c$ is a point, then by the above,
$J=]0,\min(1,\frac{1}{R})[$ is connected.
Otherwise, $J=]0, s(M_\infty(C_c)[\cup
  ]s(M_\infty(C_c),\min(1,\frac{1}{R})[$, and there are two components
  $H>c$ and $H<c$ . 
\end{proof}

\begin{rem}
  The standard definition of a foliation is slightly different from the
  one used in the theorem. 
\end{rem}

\begin{defi}
  Consider a controlled sir system defined for $t$ in an interval $
  I$.  Let $(s,i,r):I\to \RR^3$ be a solution.
Let   $M(t)=(s(t),r(t))$ . 
  The set
  $\tau \subset \RR^2$  defined by 
  $\tau:=\{M(t),\ t\in I\}$ is called a trajectory.
  \\
  If   $I=[t_0,t_1]$,
  $t_1>t_0$,  or
  $I=[t_0,+\infty[$,  $M_0\in \RR^2$  and
  $(s(t_0),i(t_0),r(t_0))=(s(M_0),i(M_0),r(M_0))$, 
 then $\tau$ is called   
 a trajectory with initial condition $M_0$.
 \\
  The
  trajectory is constant if $\tau$ is reduced to a point $M_{\tau}$.
  A constant trajectory is stable ( or equivalently, the point $M_\tau
  $ is stable) if for every $\epsilon>0$, there
  exists $\delta$ such that for every $M_0$ with $||M_0-M_\tau||<\delta$,
  the solution $M(t)$  with initial condition $M_0$  at
  $t=t_0$ satisfies $||M(t)-M_\tau||<\epsilon$ for all $t\geq t_0$. 
\end{defi}

The following proposition characterises the trajectories in terms of
the leaves of the foliation. Essentially, the leaves are unions of
trajectories.  
\begin{prop}
\label{prop:trajectoriesInFoliations}
  Suppose that the function $\beta$ in the sir-controlled
  system is constant. Let $R=\frac{\beta}{\mu} $ and consider the
  foliation associated to $R$. \\
  \begin{enumerate}
  \item Any trajectory is included in a $R$-leaf $C_c$. 
      \item 
    A subset $\tau\subset \RR^2$ is a non constant trajectory if and
    only if the following conditions are satisfied:
    \begin{itemize}
    \item There exists a $R$-leaf $C_c$ with $\tau \subset C_c$,
    \item $0 \notin i(\tau)$,
    \item The set $r(\tau)$ is an interval $I_r$ non reduced to a point. 
    \end{itemize}
  \item 
    A trajectory with initial condition $M_0$ is constant if and only
    if $i(M_0)=0$. A constant trajectory is stable if and only if
    $R<1$ or
    ($R\geq 1$ and $r(M_0)\geq 1-\frac{1}{R} $). 
  \item Let $\tau=\{M(t)\}$ be a non constant trajectory.
Recall that  $(s_\infty,r_{\infty})=\lim_{t
        \to \infty}M(t)$ is well defined. 
      Let   $c=\ln(s_{\infty})+Rr_\infty$ and let $C_c$ be the associated leaf. 
     Then $M_{\infty}(C_c)=(s_\infty,r_{\infty})$.  In other words,
     the two notions of point at infinity  ( in the foliation and in
     the differential equation) co\"\i ncide. 
  \end{enumerate}
\end{prop}
\begin{proof}
It follows from the sir-system that the equation $\frac{s'}{s}
=\frac{-\beta}{\mu} r'$. By integration, it follows that
$H(s,r)=\ln(s)+Rr$ is constant on a trajectory. This proves the first
item.

For item 4), we have seen that $s_{\infty}+r_{\infty}=1$ and that $C_c$
is closed. It follows that $(s_\infty,r_\infty) \in (C_c \cap
(i=0))$. Now,  $C_c$ has one ore two points $M$ with $i(M)=0$, namely
  $M_\infty(C_c)$ and maybe $M_{init}(C_c)$. However, $r(t)$ is
  strictly increasing, thus $M_{init}$ 
  defined by the minimality of $r$ cannot be equal to
  $(s_\infty,r_\infty)$. By elimination,  $(s_{\infty},r_{\infty})=(M_\infty(C_c))$. 

For the point $2)$, it is easy to show that a trajectory $\tau$
satisfies these three conditions. Conversely, suppose that $\tau$
satisfies the three conditions. We consider first the case with 
$I_{r}=[a,b]$ is a closed interval. Remark that any point $M$ of $\tau$
is determined by the value $r(M)$, namely $M=(e^{c-Rr(M)},r(M))$. 
Let $M(t)$ be the solution defined
on $I=[0,+\infty[$ and with initial condition
$M(0)=(e^{c-Ra},a)$. Then $b\leq r(M_\infty(C_c))=r_\infty$ by
construction of  $M_\infty(C_c)$ which has the maximum possible $r$ on a
$R$-leaf. Moreover, $b<
r_{\infty}$, otherwise $\tau$ contains $M_\infty$ which is a point
with $i=0$. Thus, by intermediate value theorem, there exists $t_0\in [0,+\infty[$ with
$r(M(t_0))=b$. The set $\{M(t),\ t\in [0,t_0]\}$ is $\tau$. Thus we
have proved item $2)$ when $I_r=[a,b]$. If $I_r=]a,b[$
is open, or semi-open, we glue the solutions defined on $[a_n,b_n]$
with $a<a_n<b_n<b$.

For item 3), it is obvious that the trajectory is constant iff
$i(M_0)=0$. We prove the stability statements in the case $R> 1$
( the case $R\leq 1$ is  easier). When $R>1$, 
there are two type of points in $T$ satisfying $i=0$,
namely the points $M$ of the form $M_\infty$ which are characterised
by $r(M)\geq 1-\frac{1}{R} $, and the points $M$ of the form $M_{init}$
which are characterised  by  $r(M)\leq 1-\frac{1}{R} $.
Thus, to prove
item 3) when $R>1$, we need to show that the points $M_{\infty}$ are stable and
that the remaining points are unstable. 
( If $R\leq 1$, all points with $i=0$ are points of the form
$M_{\infty}$ ).

Let $M_0$ be any point in $T$, $M_0(t)$ the solution of the sir-system
with initial condition $M_0$, and $M_{0}(\infty)=(s_\infty,r_\infty)$ be the limit
at infinity of  $M_0(t)$.

Let $D \subset T$ be the locus defined by $i=0$, let $D^+\subset D$
be the set of points $(s,r)$, with $r\geq 1-\frac{1}{R} $, $D^-\subset D$
defined by $r\leq 1-\frac{1}{R} $, $T'\subset T$ defined by $i\neq 0$.
Thus $T=T'\cup D^+ \cup D^-$ and $D^+\cap D^-$ is the point
$N=(\frac{1}{R} , 1-\frac{1}{R} )$.  In
particular, a function $L:T\to T$ is continuous on $T'\cup D^+$ if
its restrictions $L_1:T'\cup D^+\to T$ and  $L_2:D^-\to T$ are both
continuous (Reason: $L$  is continuous on the interior of $T'\cup D^+$
and on $N$ which is in both  sets of the covering $T=(T'\cup D^+)\cup
D^-$ ). We want to apply this remark 
when $L:T\to T$, $M_0\mapsto M_0(\infty)$ is the limit function.
The restriction of $L$ to $D^-$ is identity. Thus we need only 
prove that the restriction to $T'\cup D^+$ is continuous. 
The function
$M_0\mapsto H(M_0(\infty))$ is continuous because
$H(M_0(\infty))=H(M_0)$. Now the restriction $H_{D^+} $ of $H$ to
the interval $D^+$ is
continuous and injective, thus an homeomorphism on its image, with
inverse $H_{D^+}^{-1}$. By composition, 
the limit function  $L:T'\cup D^+\to D^+$, $M_0\mapsto H_{D^+}^{-1}(H_{D^+}(M_0(\infty)))=M_{0}(\infty)$ is
continuous on $T'\cup D^+$.

By continuity of
$L$ on $T'\cup D^+$, if $M_0\in T$ is close to  a fixed $M_\infty\in D^+$,
$L(M_0)=M_{0}(\infty)$  is
close to $L(M_\infty)=M_\infty$.  Since for every $t$,  $M_0(t)$ is included
in the rectangle with diagonal $[M_0,M_0(\infty)]$, it follows
that $M_0(t)$ is close to $M_\infty$ too. This proves the stability of
$M_\infty$. 

A point $M_{init}$ in $D\setminus D^+$ satisfies
$r(M_{init})<1-\frac{1}{R} $. It is unstable since for any choice of the initial
condition $M_0$ close to $M_{init}$ with
$i(M_0)>0$, the limit
$M_0(\infty)$ satisfies $r(M_0(\infty)) \geq 1- 1/R$. Thus for $t$ large,
$M_0(t)$ is not close to $M_{init}$. 
\end{proof}

\subsection{Optimal mitigations}
\label{sec:optiimal-mitigations}

The space of all possible mitigations is infinite dimensional. Among
all possible mitigations, what is the optimal value $s_{opt}$ maximising
  the number $s_{\infty}$ of never infected  people ?  The goal of this section is to
  compute this optimal value
  ( Theorem  \ref{thm:optimum_absolu_pour_une_strategie}). This result
  holds for finite time controls, \textit{i.e.} life
  goes back to normal after some time and mitigations may be arbitrarily
  long but don't last forever . The proof is a direct consequence of
  our study of the stability of fixed points.

\begin{defi} \label{def:sopt}
  A controlled sir system has a finite time control $\beta$ if  
  $\beta(t)=\beta_0$ for $t$ large enough.
  \\
  We fix an initial point $M_0=(s_0,r_0)$ with $i(M_0)\neq 0$, and we denote by 
  $s_{\infty}(\beta)$ the limit $s_{\infty}$ of the sir system. 
   We define
  $s_{opt}=\sup_{\beta}s_\infty(\beta)$, where $\beta$ runs through
  all finite time controls.  
\end{defi}

\begin{thm}\label{thm:optimum_absolu_pour_une_strategie}
  With the notations of definition \ref{def:sopt}, 
  \begin{itemize}
  \item If $s_0\geq s_{herd}$, then $s_{opt}=s_{herd}$. 
  \item If $s_0\leq s_{herd}$, then $s_{opt}=s_0$. 
  \end{itemize}
\end{thm}
\begin{proof}
  First, consider the case $s_{0}\leq s_{herd}$. Since $s(t)$ is decreasing and $s(0)=s_0$, we have for any $\beta$,
  $s_{\infty}(\beta)\leq s_0$ and then $s_{opt}\leq s_0$. If we take
  $\beta=0$, then $\lim_{t\to \infty}M(t)=(s_0,1-s_0)$,  which is
  a stable limit by proposition \ref{prop:trajectoriesInFoliations},
  item 3. Let
  now $\beta'$ defined from $\beta$ by relaxation after some time
  $t_r$ to get a finite time strategy:  $\beta'(t)=0$ for  $t\in [0, t_r]$ and
  $\beta'(t)=\beta_0$ for $t>t_r$. If
  $t_r$ is large,
  then $M(t_r)$ is as close as we want to $(s_0,1-s_0)$. By stability,
  it follows $M_{\infty,\beta'}$ is as close as we want to $(s_0,1-s_0)$. This shows that
  $s_{opt}\geq s_0$. 

  Suppose now that $s_0\geq s_{herd}$. For any finite time strategy
  $\beta$, there exists a time $t_0$ such that
  $\beta(t)=\beta_0$ for $t\geq t_0$. For $t\geq t_0$, we can apply
  proposition \ref{prop:trajectoriesInFoliations}, item 4 and Theorem
  \ref{thm:foliation}, item 7, to conclude that $s_{\infty}(\beta)\leq
  s_{herd}$. Thus, $s_{opt}\leq s_{herd}$.
  \\
Take $R_1$ in order to have the relation: 
$\ln s_0 +R_1r_0= \ln s_{herd}+R_1r_{herd}$. Explicitly: 
$R_1= \frac{\ln s_0 -\ln s_{herd}}{r_{herd}-r_0}$. Observe that
 we are in a realistic situation :
$R_1\geq 0$ and  $R_1\leq R_0$. This  
 follows from the hypothesis $R_0s_0\geq 1$ and from the inequalities
 $r_{herd}-r_0>s_0-s_{herd}$,  $\ln x < x-1$ for $x>1$.
  Take the solution $(s(t),i(t),r(t))$ of the sir system with
  $\beta(t)$ a constant function equal to 
  $\beta_1= \mu R_1$ and with initial condition $(s_0, 1-s_0-r_0,
  r_0)$. The choice of $R_1$ implies that $s_{\infty}=s_{herd}$ and
  $r_{\infty}=r_{herd}$. Now, we relax at a time $t_r$, \textit{i.e.} we consider
  $\beta'(t)=\beta(t)$ for $t\leq t_r$ and $\beta'(t)=\beta_0$ for
  $t< t_r$.  Then $M(t_r)$ is arbitrarily close to $M_{herd}$ if $t_r$ is large
  enough, and since $M_{herd}$ is a stable
  point, $M_{\infty}(\beta')$ remains as close as we want to $M_{herd}$.  This shows
  that $s_{opt}\leq s_{herd}$.
\end{proof}  

\begin{rem}
  The theorem states that there exists a mitigation with $s_{\infty}$
  as close as we want to $s_{opt}$, but there is no finite time strategy with
  $s_{\infty}(\beta)=s_{opt}$. 
\end{rem}

\subsection{Energy of the system}
\label{sec:energy-system}

In this section, we build an analogy with the mechanical systems.
We introduce the energy function $h_0$ which is a renormalisation of $H$
when $R=R_0$. We show that the number of finally infected people $r_{\infty}$
is an increasing function of  $h_0$ (Theorem \ref{thm:damageIncreasesWithEnergy}).
This is analogous
to the fact that a collision between highly energetic objects implies
serious damage.
The breaking mechanical power  ( in the usual physical understanding ) of the mitigation
is then $\frac{dh_0}{dt} $. 
The formula of Theorem \ref{thm:breakingPower} shows that this power is negative
for a mitigation.  Thus a mitigation can be seen as a breaking which lowers the
energy and the final damage.  When $i=0$,
the power is zero. Thus breaking when $i$ is small is not efficient. 

\begin{defi}
  Recall the function $H(s,r)=\ln(s)+Rr$.
  When the constant $R$ is equal to the basic propagation number
  $R_0$ of an epidemic, we use the specific notation $H_0(s,r)=\ln(s)+R_0r$ for the
  function $H$. 
  We let $h_0(s,r)=-H_0(s,r)-\ln(R_0)+R_0-1$. The quantity $h_0(s,r)$ is called the energy
  at point $(s,r)$. 
\end{defi}

\begin{prop}If $R_0\geq 1$, 
  the energy $h_0$ is non negative  on the triangle $T$ and
  $\min_{(r,s)\in T}h_0(r,s)=0$. Moreover $M_{herd}$ is the
  unique point with $h_0(M_{herd})=0$. If $R_0<1$, then $(1,0)$
  is the unique point of $T$ where the energy is minimal.   
\end{prop}
\begin{proof}
  We have proved in Theorem \ref{thm:foliation}, first paragraph, that
   $H_0$ has a unique maximum on $T$ located at $M_{herd}$ or $(1,0)$
   depending on the value of $R_0$.  The result
  for $h_0$ follows. 
\end{proof}

\begin{thm}\label{thm:damageIncreasesWithEnergy}
  Let $M_0=(s_0,r_0) \in T$ with $i(M_0)\neq 0$. Let $M(t)$ be the
  solution of the non controlled sir-system with $\beta(t)=\beta_0$ and initial
  condition $M_0$. Let $M_{\infty}=(s_\infty,r_\infty)$ be the limit
  at infinity. Then
  \begin{itemize}
  \item $M_{\infty}$ depends only on the energy $h_0(M_0)$, \textit{i.e.} there
    exists a
    function $g:[0,+\infty[\to T$ with $M_\infty=g(h_0(M_0))$. 
  \item $r_{\infty}$ is an increasing function of the energy
    $h_0(M_0)$. The relations $r_{\infty}>r_{herd}$ and
    $\ln(1-r_\infty)+R_0r_\infty=\ln(s(M_0))+R_0r(M_0)$ characterise
    $r_\infty$. 
  \end{itemize}
\end{thm}
\begin{proof}
  We have shown in Theorem \ref{thm:foliation} and Proposition
  \ref{prop:trajectoriesInFoliations} that in the case $i>0$,  the
  trajectory is non constant, and $M_\infty$ is characterised
  by $M_\infty\in C_{h_0(M_0)}$ and $r(M_\infty)\geq r_{herd}$. This
  proves that $M_\infty$ depends only on $h_0(M_0)$. The energy of
  the point $M_\infty=(s_\infty,r_\infty)$ is up to a constant
  $-\ln(1-r_\infty)-R_0r_\infty$.  This energy is a strictly increasing
  function of $r_\infty$ on the domain $r_\infty\in [r _{herd},1[$. 
\end{proof}

\begin{thm}
  \label{thm:breakingPower}
  Let $\beta(t)$ be a control and $M(t)$ a solution of the associated
  sir-system. Let $R(t):=\frac{\beta(t)}{\mu}$ be the propagation
  number induced by the mitigation and let
  $h_0(t)$ be the energy at time $t$.
  Then $\frac{dh_0}{dt} = \mu(R(t)-R_0)i(t)$. For a
  mitigation at time $t$, the instantaneous power $\frac{dh_0}{dt}$ 
  is negative.
\end{thm}
\begin{proof}
  By definition of $h_0(t)$, we have $h_0'(t)=-\frac{s'(t)}{s(t)}
  -R_0r'(t)$. By the sir system, $\frac{s'(t)}{s(t)} +R(t)r'(t)=0$
  and $r'(t)=\mu i(t)$. The formula for $h_0'(t)$ follows. 
A mitigation at time $t$ satisfies ( by definition ) 
  $R(t)<R_0$ thus the power is negative. 
\end{proof}

\subsection{Trajectories and solutions}
\label{sec:traj-solut}
Starting with a parameterised solution $t \mapsto (s(t),i(t),r(t))$,
$\RR \to \RR^3$ of a controlled
sir-system, we can forget the time $t$ keeping only the non
parameterised curve $C\subset \RR^3$ which is the image of the
solution. The time parametrisation
is apparently lost by this reduction, but this is not true. 
We can recover the time $t$ ( up to translation) and even the control
$\beta(t)$ that induces the trajectory $C$
( Theorem \ref{thm:computationOfImplicitBeta}). This fact will be used in the next
sections to compare how coercive  are various strategies :  We will recover $\beta(t)$
from the trajectories associated to the strategies.

We work as before in the $(s,r)$-plane with $C\subset \RR^2$ instead of $C\subset
\RR^3$ since $i=1-s-r$. At the level of the differential equation, we
are interested in the system $(**)$ formulated using only the
variables $r$ and $s$ :
\begin{displaymath}
  (**)\begin{cases}
    s'&=-\beta(t) s (1-r-s)\\
    r'&=\mu (1-r-s)
  \end{cases}
\end{displaymath}

\begin{defi}
%  A trajectory of the controlled sir system $(*)$ is a curve  $C\subset \RR^3$ image
 %of $t\mapsto (s(t),i(t),r(t))$, where $(s,i,r)$ is a solution of $(*)$. \\
  An infinite trajectory of the system $(**)$ is a set  $C\subset \RR^2$ image
  of $t \mapsto (s(t),r(t))$, where $(s,r)$ is a solution of $(**)$
  defined for $t\in [0,+\infty[$ and satisfying $(s(0),r(0))\in T$.  \\
\end{defi}

% \begin{prop}
%   There is a one-to-one correspondance between the solutions of $(*)$
%   and $(**)$. The  trajectories of $(**)$ are the projections of the
%   trajectories of $(*)$, and these projections are diffeomorphisms. 
% \end{prop}
% \begin{proof}
%   Easy and left to the reader using $i(t)=1-s(t)-r(t)$. 
% \end{proof}

\begin{rem}\label{rem:otherModels}
  The equations in $(**)$ are the equations we use throughout this
  article. As long as a model carries explicitly or implicitly these
  equations, the results of the present paper apply.
  This is the case for the sird model with deaths which boils down to a simple
  sir model if the deaths and the recovered are gathered in a unique
  compartment.  We also expect that the same approach and methods would
  lead to similar results as long as
  we can eliminate time $dt$ from the equations to get an equation
  with separate variables $\frac{ds}{s} =-R_0dr$ leading to the same associated foliation.  This the case for
  instance for the seir model without vital dynamics.
\end{rem}

\begin{thm}
  \label{thm:trajectoriesForBetaNonConstant}
  \label{thm:computationOfImplicitBeta}
  Let $C\subset T$ be a subset.  
  The following conditions are equivalent.
  \begin{enumerate}
  \item $C$ is a non constant infinite trajectory of a sir-controlled system
    for some control function $\beta(t)$,
    \item
      \begin{itemize}
      \item $r(C)$ is a semi-closed interval $[r_0,r_{\infty}[$ 
      \item There exists a function $\tilde s:[r_0,r_\infty[\to \RR$
        continuous, piecewise $C^1$, decreasing, satisfying $\tilde
        s(r)+r<1$,  whose graph $\Gamma_{\tilde s}$
        is  the set $C$,
      \item  The limit $s_\infty:=\lim_{r \to r_\infty}\tilde
        s(r)$ satisfies $s_\infty+r_\infty=1$.   
      \end{itemize}
    \end{enumerate}
    Moreover,
    we have the formulas $t=\int_{r_0}^r \frac{dr}{\mu(1-r-\tilde
      s(r))}$ and   $\beta(t)=\frac{-\mu s'(t)}{s r'(t)} $.
\end{thm}
\begin{proof}
  We prove first that $1\Rightarrow 2$. By Theorem
  \ref{thm:descriptionOfTheSolutions}, a solution $s(t)$ is decreasing
  and piecewise $C^1$ (with respect to t).    The sir-equation  $r'=\mu i$ implies that $r=r(t)$ is  an increasing ${\emph C}^1$ diffeomorphism 
 from $t\in [0,\infty[$ onto $r\in [r_0, r_{\infty}[$.  By inversion we obtain
 a function $t=t(r)$ which is $C^1$ and strictly increasing.  The function $s$ of
 $t$ can then be expressed with the parameter $r$  letting  $\tilde
 s(r)=s\circ t(r)$. The graph and the limits stay unchanged by this
 reparametrisation with $r$ instead of $t$ and the equalities
$s_\infty+r_\infty=1$ and $\tilde s(r)+r<1$  follow from the
 corresponding results for $s$  in Theorem
 \ref{thm:descriptionOfTheSolutions}.\\
 Conversely, we let $t=\int_{r_0}^r \frac{dr}{\mu(1-r-\tilde s(r))}
 $. Since $\tilde s$ is decreasing, it follows $1-r-\tilde s(r) \leq r_{\infty}-r$ and then 
 $t_\infty:=\int_{r_0}^{r_\infty} \frac{dr}{\mu(1-r-\tilde s(r))}
 \geq \int
 _{r_0}^{r_{\infty}}\frac{dr}{(r_{\infty}-r)\mu}=+\infty$.
  This makes $t:[r_0,r_{\infty}[\to [0,\infty[$ a strictly increasing
  $C^1$ function of $r$. We
 denote by $r=r(t):[0,\infty[\to [r_0,r_\infty[$ the inverse function
 which is $C^1$ too. We let $s(t)=\tilde s(r(t))$, $\tilde \beta(r)=-\mu
  \frac{\tilde s'(r)}{\tilde s(r)}$, and $\beta(t)=\beta(r(t))= -\mu
  \frac{\tilde s'(r(t))}{s(t)}$. We claim that $s(t),r(t)$ is the
  solution of the sir-system with control $\beta(t)$ with initial
  condition $s(0)=\tilde s(r_0)$ and $r(0)=r_0$, hence gives a
  parametrisation of $C$ as a trajectory. This is a direct
  computation: The formula for $t$ shows that $r'(t)=\frac{dr}{dt} =
  \mu (1-r(t)-s(t))$ and $s'(t)=\frac{-s(t)\beta(t)}{\mu} r'(t)=-\beta(t) s(t) (1-r(t)-s(t))$. 
 \end{proof}

\subsection{Cost functions }
\label{sec:strategies in general }

%Obviously, people prefer a lockdown of 1 month rather than 2 months.
%Similarly, people prefer the closure of restaurants rather than
%an absolute lockdown. In other words,
Some people prefer long loose interventions, while others prefer short
harsh interventions.
Mathematically, we encode the preferences with cost
functions. The cost increases when mitigations 
become more intensive or when the duration  increases.
Each person carries its own cost function.

\begin{defi} Recall that a sir model comes with a constant $\beta_0$ depending
  on biological conditions. 
  Let $\beta$ be a control with values in
  $[0,\beta_0]$. The number $\beta(t)$ is called the constraint at time $t$.
  There is no constraint 
  if $\beta(t)=\beta_0$ and there is a constraint  if
  $\beta(t)<\beta_0$.  The
  maximal constraint for the control $\beta$ is the infimum $infess
  (\beta(t))$, where $infess$ denotes the essential infimum.   
  \\
  A cost function is a decreasing function $c:[0, \beta_0] \fd \RR$  with
  $c(\beta_0)=0$.
  \\
  The duration $d(\beta)$ of a control is the Lebesgue
  measure of the locus $\beta^{-1}([0,\beta_0[)$.
  \\
  The  cost of a control $\beta$ is $c(\beta)=\int_{\RR^+}c(\beta(t)
  )dt$. A time interval $I$ such that there is no constraint outside
  $I$ is called a constraint interval. The cost function can
be computed after restriction to a constraint interval : For any
constraint interval $I$, 
\textit{i.e.} $c(\beta)=\int_{I}c(\beta(t))$.
  \\
  A control $\beta_1$ is preferable to a control $\beta_2$ (In notation :
  $\beta_1\geq \beta_2$)  if
  $r_\infty(\beta_1)\leq r_\infty(\beta_2)$ and for every 
  cost function $c$, $c(\beta_1)\leq c(\beta_2)$. The relation $\leq$  is a partial 
  preorder. A control $\beta_1\in \fxc $ is optimal  in the set of controls
  $\fxc$ if for every  $\beta_2\in \fxc$,   we have $\beta_1\geq \beta_2$. 
  A control  $\beta_1\in \fxc$ is maximal  in 
  $\fxc$ if there is no better control in $\fxc$ : for all $\beta_2\in \fxc$
 we have $\beta_1\geq \beta_2$ or $\beta_1$ and $\beta_2$ incomparable. 
  \end{defi}

\begin{rem}
  With informal words, an optimal control is a control
  preferred by
  everyone regardless of the subjective personal cost function,
  with $r_{\infty}$ minimal.
  A maximal control is a control $\beta$ which cannot be improved:
  replacing $\beta$ by $\beta'$ increases $r_{\infty}$ or makes at least one person
  dissatisfied by a higher cost. 

  If there is an optimal control in a set $\fxc$ of
  considered choices, the public policy
  is obviously chosen : an optimal control is a universal
  favourite choice for all people. If no optimal control exists, then
  there is no universal choice and political arbitrations are required. However, there are still 
  choices that can be improved universally :
  If everyone prefers $\beta_2$ than $\beta_1$, the public policy
  should reject $\beta_1$, in favour of $\beta_2$. At the end, the public
  choice should be ideally between controls that can't be universally ameliorated
  any more, \textit{i.e.}   an option is selected between the maximal controls.
\end{rem}

\begin{rem}
  We have used the essential infimum rather than the infimum in the
  definition of the maximal constraint because of the points $a_i$
  where $\beta$ may be discontinuous. 
\\
  We underline that the order for controls is different from the
  order for functions. The relation $\beta_1\leq \beta_2$ as controls 
  defined above does not imply that $\beta_1\leq \beta_2$ as usual
  functions.
\end{rem}

The following theorem gives necessary conditions to compare two
controls. If $\beta_2$ is universally preferred to $\beta_1$, then the
mitigation is shorter-lived and the maximal constraint is smaller. 
\begin{thm}\label{thm:criterion_for_maximal_control}
  Let $\fxc$ be a set of controls with the same $r_{\infty}$ : $\forall
  \beta_1,\beta_2\in \fxc$, $r_\infty(\beta_1)=r_\infty(\beta_2)$.
 Let $\beta_1,\beta_2\in \fxc$. If
  $\beta_1\leq \beta_2$, then  the duration of constraint  satisfy  $d( \beta_1)\geq d(\beta_2)$, and
  the maximum of  constraints satisfy $infess(\beta_2(t)) )\geq infess(\beta_1(t))$.  
 % $r_\infty(\beta_1)=r_\infty(\beta_2)$ for any $\beta_1,\beta_2\in \fxc$.

  % If $\beta$ in $\fxc $ is maximal, then $\forall \beta'\in
%   \fxc$,  the durations of the constraints  satisfy  $d( \beta')\geq d(\beta)$, and
%   the maxima of  constraints satisfy $infess(\beta'(t)) )\leq infess(\beta(t))$. 
% 
\end{thm}
\begin{proof}
 Consider   the cost
function $c(x)$  which is equal to $1$ if $x \in [0,\beta_0[$ and $0$
otherwise; then $c(\beta)=d(\beta)$. If $\beta_2\geq \beta_1$,  $d( \beta_1)=c(\beta_1)\geq c(\beta_2)=d(\beta_2)$.
\\
For the second inequality, we suppose ad absurdum that $infess \beta_1 >infess
\beta_2$ and we take $\beta_3$ with $infess \beta_1>\beta_3>infess \beta_2$. Then we choose the cost function 
which is $1$ in $[0, \beta_3]$ and $0$  otherwise. Then
$c(\beta_1)=0<c(\beta_2)$. 
\end{proof} 

Intuitively, if  we have two controls $\beta_1$ and $\beta_2$,
and if the constraints in $\beta_2$ are harsher and longer than
those for $\beta_1$, then the cost functions should satisfy
$c(\beta_2)\geq c(\beta_1)$. This idea is formalised 
with a suitable diffeomorphism in the next theorem. 

\begin{thm}\label{thm:dominatedImpliesInferior}
 Let $\beta_1$ and $\beta_2$ be two controls and $I_1,I_2$ two constraint
 intervals for $\beta_1$ and $\beta_2$. Suppose that
  there exists a diffeomorphism
  $\phi: I_1 \rightarrow I_2$ with $\phi'(t)\geq 1$ and
  $\beta_2(\phi(t))\leq \beta_1(t)$. Then for any cost function $c$,
  $c(\beta_2)\geq c(\beta_1)$. In particular, if
  $r_{\infty}(\beta_2)\leq r_{\infty}(\beta_1)$, then $\beta_2\leq \beta_1$.
\end{thm}

\begin{proof}
 We have $c(\beta_2)=\int_{I_2}c (\beta_2(u)) du= \int_{I_1} c(
 \beta_2(\phi(t))\phi'(t)dt\geq  \int_{I_1} c(
 \beta_1(t))dt$. 
\end{proof} 

% \begin{rem}
%   The converse of theorem \ref{thm:dominatedImpliesInferior} is not
%   true.
% \end{rem}

\subsection{Strategies with constant mitigation}
\label{sec:strat-with-const}

In this section, we consider constant mitigations, that is mitigations with a fixed level of
intervention and we investigate  the problem of their optimality : which
ones are optimal, which ones are not ?  

A first approach with constant mitigations is to minimise the burden
for a fixed duration. For instance, we may consider two
strategies for promoting remote work : The first strategy promotes
remote work every Monday for five weeks.
The second strategy promotes remote work from Monday to Friday after one month. Both
strategies have the same type of restriction, and the same duration (
5 days ). Which strategy
yields to the smallest ratio $r_\infty$ of infected people ? More generally, keeping the same
interventions and the same duration, and changing the scheduling, is there a mitigation strategy
minimising the burden of the epidemic? 

An other approach with constant mitigations could be to fix a maximal
burden $r_{\infty}$ and  a level of mitigation ( in the example above,
remote work ). Then the policy makers try
to minimise the mitigation time required to keep
$r_{\infty}$ under the chosen limit. What are the optimal 
mitigations for this question ? 

We will prove in this section that the two above problems ( minimising the
duration for a fixed burden or minimising the burden for a fixed
duration) are 
equivalent : a strategy based on a constant control $\beta$ is optimal
for one problem if and only if it is optimal for the other problem
( Theorem \ref{thm:optimal_constant_mitigations}). Such an optimal
$\beta$ always exists and it has specific properties. More precisely,
we show that for an optimal strategy,
all the mitigations must be grouped in a unique intervention. In a
sir-controlled model, several
short interventions are less efficient than a long adequately planned
intervention of the same total duration. The adequate planning
boils down to ``as soon as possible'' if the herd ratio has been
passed, and around the zone $s=s_{herd}$ otherwise.

Our argument consists in replacing the infinite dimensional
space in which $\beta$ moves with a compact space.
This reduction relies on the comparison of the time spent on parallel
$R$-leaves of a quadrilateral.  Once the ambient space is compact,  some
continuity argument can be applied. 

Technically, a constant mitigation appears as a control $\beta$
that takes two values $\beta_0$ and $\beta_1$, where $\beta_0$ is the
value of $\beta$ when no mitigation occurs, and $\beta_1<\beta_0$ is the value
of $\beta$ during the mitigation.

\begin{figure}[h]
\centering
\includegraphics[scale=.3]{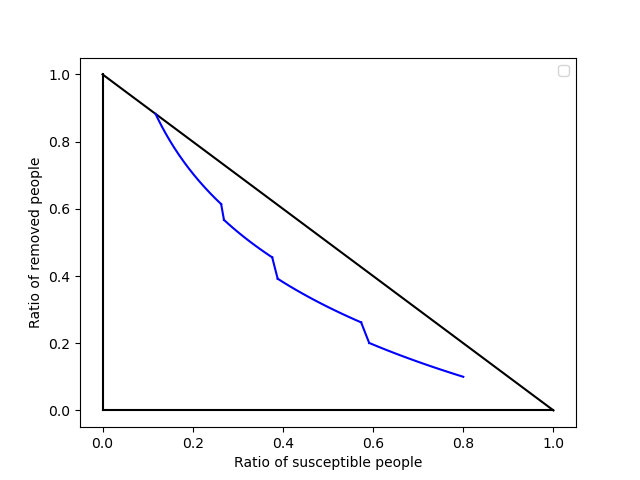}
\caption{Epidemic with a control of order 2 and $M_0=(0.8,0.1)$.}
\end{figure}

\begin{defi}
  Let $\beta_1\in [0,\beta_0[$. A control $\beta$ of order $k\geq -1$ is a
  piecewise constant function $\beta:[0,+\infty[\to \{\beta_0,\beta_1\}$ such
  that there exists elements $(a_i,b_i)_{i\in \{0\dots k\}}$ with
    $0 \leq a_0<b_0<a_1<b_1 < \dots
<a_k<b_k$, $\beta(t)=\beta_1$ on any interval 
$[a_i,b_i]$ and  $\beta(t)=\beta_0$ otherwise. A mitigation with a
control of some order $k$ is called a constant mitigation. 
Its duration  is $d(\beta)=\sum_i( b_i-a_i)$.
By convention, for $k=-1$, we have $\beta(t)=\beta_0$ for every $t$
and $d(\beta)=0$. \\
In this section, we fix once and for all an initial point $M_0\in T$, with $i(M_0)\neq 0$,  and we denote by $r_\infty(\beta)$ the
share of people finally removed for the sir-control system with
initial condition $M_0$ and control $\beta$.
\\
Recall that $H_0(s,r)=\ln(s)+R_0r$ with $R_0=\frac{\beta_0}{\mu} $.  Similarly, we let
$H_1(s,r)=\ln(s)+R_1r$, with   $R_1=\frac{\beta_1}{\mu} $.    
\end{defi}

We will consider  minimisation problems
where an optimal $\beta \in \fxc$ has to be found within various classes $\fxc$.
These classes $\fxc$ are introduced in the
following definition. 

\begin{defi}
  Let $d\geq 0,r_{\infty}\in ]r_{herd},1[$. 
  \begin{itemize}
  \item   $\fxc_{d,k}$ is the set of controls $\beta$ with constant
    mitigation of order $k$ and
    duration $d(\beta)=d$, 
  \item $\fxc_{d}:=\cup_{k\geq -1} \fxc_{d,k}$  ( note that
    $\fxc_0=\fxc_{0,-1}$  and for $d>0$, $\fxc_{d}=\cup_{k\geq 0}
    \fxc_{d,k}$ since $\fxc_{d,-1}$ is empty )
  \item $\fxc:=\cup_{d\geq 0} \fxc_d$ is the set of finite time constant controls,
  \item     $\fxc^{r_\infty}$ is the set of  $\beta\in \fxc$ with
    $r_{\infty}(\beta)=r_{\infty}$,
  \item   $\fxc^{r_\infty,k}$ is the set of  $\beta\in \fxc$
    of order $k$ with
    $r_{\infty}(\beta)=r_{\infty}$
  \item  $\fxc^{cross}$ is the set of  $\beta\in
    \fxc$ such that there exists $i$, there exists $t\in [a_i,b_i]$
    with $s(t)=s_{herd}$. In non technical terms, there is an
    alternation of mitigated and non mitigated periods, and  the herd ratio
    $s_{herd}$ is crossed during a mitigated period $[a_i,b_i]$
    rather than during a non mitigated period $]b_i,a_{i+1}[$. 
  \item  $\fxc^{imm}$ is the set of controls $\beta\in
    \fxc$ such that $a_0=0$. The exponent ``imm'' stands for
    immediate, to denote the controls that start the mitigation
    at $t=0$. 
  \item   Several exponents correspond to intersections of the classes
    above. For instance,
    $\fxc_{d}^{cross}=\fxc_{d}\cap \fxc^{cross}$, 
    $\fxc_{d,k}^{cross}=\fxc_{d,k}\cap \fxc^{cross}$, 
    $\fxc^{r_\infty,cross}=\fxc^{r_\infty}\cap\fxc^{cross}$,
    $\fxc^{r_\infty}_d=\fxc^{r_\infty}\cap\fxc_d$,
    $\fxc^{r_\infty,k,cross}=\fxc^{r_\infty,k}\cap \fxc^{cross}$. 
  \end{itemize}
\end{defi}

\begin{thm}\ \\ \label{thm:optimal_constant_mitigations}
  \begin{enumerate}
\item Let $\beta\in \fxc$, $d=d(\beta)$, and
  $r_\infty=r_\infty(\beta)$. Then $\beta $ is optimal in $\fxc_d$ if
  and only if $\beta$ is optimal in $\fxc^{r_\infty}$. 
  \item 
  For every $ d\geq 0$, there exists $\beta \in \fxc_d$ which is
  optimal in $\fxc_d$.
\item For every $r_\infty$ such that $\fxc^{r_\infty}$ is non empty,
  there  exists  $\beta \in \fxc^{r_\infty}$ which is optimal in $\fxc^{r_\infty}$.
\item  $\fxc^{r_\infty}$ is non empty if and only if $r_\infty \in
  ]r_\infty(\beta_1),r_\infty(\beta_0)]$ where
  \begin{itemize}
  \item $r_{\infty}(\beta_0)$ is the limit of $r(t)$ when
    $\beta(t)=\beta_0$ for all $t$ ( no mitigation). 
  \item $r_{\infty}(\beta_1)=\lim_{N \to \infty} r_\infty(\beta_{1,N})$
    where
    $\beta_{1,N}(t)=\beta_1$ for all $t\leq N$, and $\beta_{1,N}
    (t)=\beta_0$ for $t> N$. In intuitive terms,  $r_{\infty}(\beta_1)$ is the ratio of infected
    people after an arbitrarily long but still finite mitigation whose
    intensity is governed by $\beta_1$.
  \end{itemize}
\item We denote by $r_\infty(d)$ the quantity $r_\infty(\beta)$ when
  $\beta$ is optimal in $\fxc_d$ and by $d(r_\infty)$ the
  quantity $d(\beta')$ where  $\beta' $ is optimal in $\fxc^{r_\infty}$. 
Then the functions $d\to r_\infty(d)$ and $r_\infty \to d(r_\infty)$
are mutually inverse for $d\in [0,+\infty[$ and $r_\infty \in     ]r_\infty(\beta_1),r_\infty(\beta_0)]$.
\item If $\beta\in \fxc_d$ is optimal, then $\beta \in
  \fxc_{d,0}$. In other words, an optimal
  mitigation does not split the mitigation in parts.
\item  If $\beta$ is optimal and $s(M_0)\leq s_{herd}$, then $\beta
  \in \fxc^{imm}$.  In
  other words, the optimal mitigation starts as soon as possible if
  $s(M_0)\leq s_{herd}$.
\item   If $\beta$ is optimal and $s(M_0)\geq  s_{herd}$,  then 
  $\beta \in  \fxc^{cross}$.  In other words, the adequate
  scheduling of the mitigation is such that the vertical line
  $s=s_{herd}$ is crossed during the mitigation : the mitigation occurs
  around the zone where $s=s_{herd}$. 
\end{enumerate}

\end{thm}

We start the proof with a lightened form of the theorem,
considering only one step mitigations ( $\beta$ of order $0$) and
$s(M_0)\geq s_{herd}$.
\begin{prop}\label{pre-thm} We suppose $s(M_0)\geq s_{herd}$.  
  \begin{enumerate} 
\item Let $\beta\in \fxc_{d,0} ^{cross}$, and
  $r_\infty=r_\infty(\beta)$. Then $\beta $ is optimal in $\fxc_{d,0}^{cross}$ if
  and only if $\beta$ is optimal in $\fxc^{r_\infty,0,cross}$. 
  \item 
  For every $ d\geq 0$, there exists $\beta
  \in \fxc_{d,0}^{cross}$ which is optimal
  in $ \fxc_{d,0}^{cross}$.
\item For every $r_\infty$ such that $\fxc^{r_\infty,0,cross}$ is non empty,
  there  exists an optimal $\beta \in \fxc^{r_\infty,0,cross}$ which
  is optimal in $\fxc^{r_\infty,0,cross}$.
\item  $\fxc^{r_\infty,0,cross}$ is non empty if and only if $r_\infty \in
  ]r_\infty(\beta_1),r_\infty(\beta_0)]$ where
  \begin{itemize}
  \item $r_{\infty}(\beta_0)$ is the limit of $r(t)$ when
    $\beta(t)=\beta_0$ for all $t$ ( no mitigation). 
  \item  $r_{\infty}(\beta_1)=\lim_{N \to \infty} r_\infty(\beta_{1,N})$
    where
    $\beta_{1,N}(t)=\beta_1$ for all $t\leq N$, and $\beta_{1,N}
    (t)=\beta_0$ for $t> N$. 
  \end{itemize}
\item Denote by $r_\infty(d)$ the quantity $r_\infty(\beta)$ when
  $\beta$ is optimal in $\fxc_{d,0}^{cross}$ and denote by $d(r_\infty)$ the
  quantity $d(\beta')$ where  $\beta' $ is optimal in $\fxc^{r_\infty,0,cross}$. 
Then the functions $d\to r_\infty(d)$ and $r_\infty \to d(r_\infty)$
are mutually inverse for $d\in [0,+\infty[$ and $r_\infty \in
]r_\infty(\beta_1),r_\infty(\beta_0)]$.
\end{enumerate}  
\end{prop}
\begin{proof}
  From the definition of optimality, we have $\beta$
  optimal in $\fxc_{d,0}^{cross}$ iff $r_{\infty}(\beta)=inf_{\beta''
    \in \fxc_{d,0}^{cross}}r_{\infty}(\beta'')$. 
Indeed, all $\beta\in \fxc_d$ have the same cost $c(\beta_1)d$ for
every cost function $c$. Thus $\beta<\beta'$ as controls if and only iff
$r_\infty(\beta)>r_\infty(\beta')$. Similarly, $\beta$
  optimal in  $\fxc^{r_\infty,0,cross}$ iff $d(\beta)=inf_{\beta''
    \in \fxc^{r_\infty,0,cross}}d(\beta'')$.

  Suppose that $\beta$ is not optimal in $\fxc^{r_\infty,0,cross}$, then
  there exists $\beta'$ with $d(\beta')<d(\beta)$ and
  $r_\infty(\beta')=r_\infty$. The control $\beta'$ has value
  $\beta_1$ on an interval $[a_0,b_0]$ with $b_0<a_0+d(\beta)$. Define
  $\beta''$ by $\beta''(t)=\beta_1$ for $t\in [a_0,a_0+d(\beta)]$ and
  $\beta''(t)=\beta_0$ otherwise. Since there is a longer mitigation
  in $\beta''$ than in $\beta'$, it follows that
  $r_{\infty}(\beta'')<r_{\infty}(\beta')=r_\infty$. Thus
  $\beta$ is not optimal in  $\fxc_{d,0}^{cross}$ since $\beta''$ is a
  better choice.

  Suppose conversely that $\beta$ is not optimal in
  $\fxc_{d,0}^{cross}$. Then  there exists $\beta'$ with $r_\infty(\beta')<r_\infty(\beta)$ and
  $d(\beta')=d$.  The control $\beta'$ has value
  $\beta_1$ on an interval $[a_0,a_0+d]$. Define   $\beta''_u$
  depending on $u\in [0,d]$ by $\beta''_u(t)=\beta_1$ for $t\in [a_0,a_0+u]$ and
  $\beta''(t)=\beta_0$ otherwise. For $u=0$, there is no mitigation thus
  $r_{\infty}(\beta''_u)>r_{\infty}(\beta)$. For $u=d$,
  $\beta''_u=\beta'$, thus the
  opposite inequality holds. By the intermediate value theorem, there
  exists $u_0\in ]0,d[$ with
  $r_\infty(\beta''_{u_0})=r_\infty(\beta)$.   If $\beta''_{u_0}\in
  \fxc^{cross}$, then
  $\beta$ is not optimal in  $\fxc^{r_\infty,0,cross}$ since $\beta''_{u_0}$ is a
  better choice. If $\beta''_{u_0}\notin
  \fxc^{cross}$, then we replace $\beta''_{u_0}$ by $\beta'''\in
  \fxc^{cross}$ using proposition \ref{reduction_to_cross}. This concludes the proof of item 1.

  Now we prove item 2. Let $\beta_{no}=\beta_0$ be the constant control
  corresponding to no mitigation and let $t_{herd}$ be such that
  $s(t_{herd})=s_{herd}$ for the sir-system defined by
  $\beta_{no}$. The existence of $t_{herd}$ follows from the
  hypothesis $s(M_0)\geq s_{herd}$ and from the equality
  $s_{\infty}<s_{herd}$ (Theorem
  \ref{thm:damageIncreasesWithEnergy}). 
    For $u\in
 [0,t_{herd}]$, define $\beta_u$ by $\beta_u(t)=\beta_1$ for $t\in
  [u,u+d]$, and $\beta_u(t)=\beta_0$ otherwise.  Let $s_u(t)$ be the function $s$ associated to the
  sir-system with control $\beta_u$. Then $\beta_u\in C^{cross}$ if
  and only if $s_u(u+d)\leq s_{herd}$. By continuity in $u$, this is a
  closed condition on the parameter $u\in [0,t_{herd}]$, thus $u$
  lives in a compact $K$. It follows
  that the elements $\beta \in C^{cross}_{d,0}$ are parameterised by  
 an element  $u\in K$.  The map $K\to \RR$, $u\mapsto
 r_{\infty}(\beta_u)$ is a continuous function on $K$, hence it admits
 a minimum.  This proves item 2. 

 Item 4 is easy. 
  
  For every $d\in D= [0,+\infty[$, the set
  $\fxc_{d,0}^{cross}$ contains an optimal control $\beta$ by item 2.
  Let $B\subset
  [0,1]$ be the set of elements $r_{\infty}$ such that
  $\fxc^{r_\infty,0,cross}$ contains an optimal control. Let $\beta
  \in \fxc$ be a control of order $0$. If $\beta$
  is optimal ( equivalently in  $\fxc^{r_\infty(\beta),0,cross}$ or
  $\fxc_{d(\beta),0}^{cross}$ ), then $d(\beta)\in D$, $r_\infty
  (\beta) \in B$, and
  $d(r_\infty(\beta))=d(\beta)$ and
  $r_{\infty}(d(\beta))=r_\infty(\beta)$.
  This proves that the functions $r_{\infty}$ and $d$ in item 5) are
  mutually inverse one-to-one correspondences between $D=[0,+\infty[$
  and  $B$. Since $r_{\infty}(d)$ is a continuous strictly decreasing
  function of $d$, we have $B=]\lim_{d\to
    \infty}r_\infty(d),r_\infty(0)]$. A mitigation of
  duration $0$ means no mitigation, thus
  $r_\infty(0)=r_\infty(\beta_0)$. The equality $\lim_{d\to
    \infty}r_\infty(d)=r_\infty(\beta_1)$ holds by definition of
  $r_\infty(\beta_1)=\lim_{N\to \infty}
  r_\infty(\beta_{1,N})$ since $r_\infty(N)\leq r_\infty(\beta_{1,N})$
  and $r_\infty(d)\geq r_\infty(\beta_{1,N})$ for $N$ large enough. This concludes the proof of
  item 5 and item 3. 
\end{proof}

To prove Theorem \ref{thm:optimal_constant_mitigations},
the main point now is to show that we can reduce to
the simple version we proved
in Proposition \ref{pre-thm}. That is, we 
need to show  that an optimal 
control $\beta$ corresponds to a one step mitigation
and is in
$\fxc^{cross}$. This will be done in Proposition
\ref{reduction_to_cross}.  
Proposition  \ref{reduction_to_cross} requires some lemmas
relative to $R_0-R_1$ quadrilaterals that we define and study now. 

\begin{defi}
  A $R_0-R_1$ quadrilateral is a set of $4$ points $a,b,c,d\in T$ with
  \begin{itemize}
  \item $a,b$ are on a common $R_0$-leaf $C_{ab}$, 
  \item $c,d$ are on a common $R_0$-leaf $C_{cd}$,
  \item $a,d$ are on a common $R_1$-leaf $C_{ad}$,
  \item $b,c$ are on a common $R_1$-leaf $C_{bc}$,
  \end{itemize}
\end{defi}

The above definition does not fix the order of the 4 points
$a,b,c,d$ of the quadrilateral.
The following lemma indicates how the order can be fixed. 
\begin{lm}\label{lm:orderInTheQuadrilatere}
  In a $R_0-R_1$ quadrilateral, the following conditions are
  equivalent:
  \begin{enumerate}
  \item   $r(a)=max\{r(a),r(b),r(c),r(d)\}$ 
  \item     $r(c)=min\{r(a),r(b),r(c),r(d)\}$. 
  \item $H_1(a)=H_1(d)<H_1(c)=H_1(b)$ and
    $H_0(a)=H_0(b)>H_0(d)=H_0(c)$. 
  \end{enumerate}
\end{lm}
\begin{proof}
  For $p=(s_p,r_p)$, $q=(s_q,r_q)$, we have
  $(H_0(p)-H_0(q))-(H_1(p)-H_1(q))=(R_0-R_1)(r_p-r_q)$.  Since
  $R_0>R_1$, the
  equivalences $1\Leftrightarrow 3$ and $2 \Leftrightarrow 3$ follow easily. 
\end{proof}

\begin{lm}[Completion of a triangle to a quadrilatere. ]
  \label{lm:completionOfATriangle}
  Let $a,b,c\in T$ with $H_0(a)=H_0(b)>H_0(c)$ and $H_1(b)=H_1(c)>H_1(a)$.
  Then there exists $d\in T$ such that $a,b,c,d$ is a
  $R_0-R_1$ quadrilateral satisfying the ordering of lemma
  \ref{lm:orderInTheQuadrilatere}. 
\end{lm}
\begin{proof}
  By the intermediate value theorem, there exists $e\in T$ with
  $r(e)=r(c),s(e)\leq s(c)$, and $H_1(e)=H_1(a)$. Then $H_0(e)\leq
  H_0(c)<H_0(a)$. The $R_1$-leaf containing $a$ and $e$ is connected
  by Theorem \ref{thm:foliation}. When joining $e$ to $a$ in the $R_1$-leaf, we find by the
  intermediate value theorem a point $d$ with $H_0(d)=H_0(c)$. 
\end{proof}
The previous lemma recovered $d$ from $a,b,c$. 
We have an other completion lemma which has a similar proof, which
recovers $b$ from $a,c,d$ when $s(c)<s_{herd}$. 

\begin{lm}[Completion of a triangle to a quadrilatere. ]
  Let $a,c,d\in T$ with $H_0(a)>H_0(c)=H_0(d)$ and $H_1(c)>H_1(a)=H_1(d)$.
  Suppose moreover that $s(c)<s_{herd}$. Then there exists $d\in T$ such that $a,b,c,d$ is a
  $R_0,R_1$-quadrilateral satisfying the ordering of lemma
  \ref{lm:orderInTheQuadrilatere}. 
\end{lm}
\begin{proof}
  Sketch. From $c$, we draw a vertical line towards the line $i=0$ to
  build a point $e$. We then conclude as in the proof of lemma \ref{lm:completionOfATriangle}. 
\end{proof}

The following lemma compares the duration of the mitigation on two opposite
edges of a quadrilateral. 
\begin{lm}\label{lemmaTimeOnQuadri}
Let $a,b,c,d$ be a quadrilateral satisfying the ordering of lemma
  \ref{lm:orderInTheQuadrilatere} and suppose that $s(a)\geq s_{herd}$. 
Consider the oriented curves
\begin{itemize}
\item $C_{ba}$, $C_{cd}$ the $R_0$-curves joining $a$ to $b$ and $c$
  to $d$
\item $C_{da}$,$C_{cb}$ the $R_1$-curves joining $d$ to $a$ and $c$ to
  $b$. 
\end{itemize}
Then the time spent on $C_{cb}$ is strictly longer  than the time spent on
$C_{da}$.
\\
If the hypothesis $s(a)\geq s_{herd}$ is removed and replaced with
$s(c)\leq s_{herd}$, then the opposite conclusion holds : strictly less time is
spent on $C_{cb}$ than on $C_{da}$. 
\end{lm}
\begin{proof}
By Theorem  \ref{thm:breakingPower}, the time spent on $C_{cb}$ is
$\int_{C_{cb}}dt(M)=\int_{C_{cb}}\frac{dH_0(M)}{(\beta_0-\beta_1)i(M)}$.
Thus, up to a positive multiplicative constant, the time is measured
by integrating the differential form $\frac{dH_0}{i} $. 
Denote by $\phi:C_{cb}\fd C_{da}$ the diffeomorphism which sends the point $M$ of $C_{cb}$ to  the point $M'$  of $C_{da}$ with $H_0(M)= H_0(M')$. By construction, 
$H_0=H_0\circ \phi$. This implies $dH_0= \phi ^{*}(dH_0)$. Also, using
$\phi$ as a change of variable,
%the integral $\int _{C_{cb}}\frac{dH_0}{i}$ is equal to the integral 
%$\int _{C_{da}} \phi^{*}( \frac{dH_0}{i})$. By the above, we obtain:
the time spent on $C_{da}$ is 
$$ \int_ {C_{da}}\frac{dH_0(M')}{(\beta_0-\beta_1)i(M')} =
% \int_ {C_{cd}}\frac{dh_0(\phi(M))}{i(\phi(M))} =
\int_{C_{cb}}\frac{dH_0(M)}{(\beta_0-\beta_1)i(\phi (M))}.$$
The result follows since $i(\phi(M))$ is larger than $i(M)$ by
proposition \ref{prop:i_increasing_then_decreasing}.
\end{proof}

\begin{prop}\label{reduction_to_cross}
  Let  $\beta \in \fxc^{r_\infty}\setminus  \fxc^{r_\infty,cross,0}$ and suppose
  that   $s(M_0)\geq s_{herd}$. Then there exists
  $\beta'\in  \fxc^{r_\infty,cross,0}$ with $d(\beta')< d(\beta)$. 
\end{prop}
\begin{proof}
Suppose that $\beta \notin \fxc^{cross}$.   We want to construct
$\beta_1\in \fxc^{cross}$ with the same $r_{\infty}$ and a smaller
duration using lemma \ref{lemmaTimeOnQuadri}. The trajectory
of the solution associated to $\beta$ is characterised by a sequence
of points $M_0=M(0), N_1=M(a_0), M_1=M(b_0),N_2=M(a_1),M_2=M(b_1),
\dots, N_{k+1}=M(a_k),M_{k+1}=M(b_k)$ where:
\begin{itemize}
\item $N_i,M_i$ are joined by a $R_1$-leaf, 
\item $M_i,N_{i+1}$ are joined by a $R_0$-leaf. 
\end{itemize}
Since $\beta\notin \fxc^{cross}$,
$s(N_1)<s_{herd}$ or $s(M_1)>s_{herd}$.
By symmetry, we consider the ``right'' case and we suppose
$s(M_1)>s_{herd}$.  We take $i$ maximum such that 
$s(M_i)>s_{herd}$.   There
exists $t$ with $s(t)=s_{herd}$. We apply lemma
\ref{lm:completionOfATriangle} to
$a=M(t)$, $b=M_i$, $c=N_i$ which gives a point $d$ such that $a$ and
$d$ are on common $R_1$-leaf, and $d$ is on the $R_0$-leaf common to
$c$ and $M_{i-1}$. We consider $\beta_1$ the control associated to the
trajectory with characteristic points $M_0,N_1,M_1,\dots,
N_{i-1},M_{i-1},d,a=M(t),N_{i+1},M_{i+1}...$. By 
lemma \ref{lemmaTimeOnQuadri}, we have $d(\beta_1)<d(\beta)$. 

Replacing $\beta$ by $\beta_1$, we may suppose that $\beta\in
\fxc^{cross}$. If $k=0$, we are done so we suppose $k>0$ and we will construct
$\beta_2\in \fxc^{cross}$ with $r_\infty(\beta_2)=r_\infty(\beta)$, 
$d(\beta_2)<d(\beta)$ and $\beta_2$ has a smaller $k$ than $\beta$.
We take again the initial notations where $\beta$ is characterised by the
points $M_0,N_1,\dots, N_{k+1},M_{k+1}$. 
Since  $\beta \in \fxc^{cross}$, there exists $i\geq 1$ with $s(N_i)\geq
s_{herd}$ and $s(M_i)\leq s_{herd}$. Since $k+1\neq 1$, we have $i\neq 1$ or
$i\neq k+1$. By symmetry, we suppose $i>1$. We apply lemma   \ref{lm:completionOfATriangle} to
$a=N_i$, $b=M_{i-1}$, $c=N_{i-1}$,  which gives a point $d$. The point 
$d$ is on the $R_1$-leaf containing $N_i$,$M_i$, and on the $R_0$-leaf common to
$N_{i-1}$  and $M_{i-2}$.   We consider $\beta_2$ the control associated to the
trajectory with characteristic points $M_0,N_1,M_1,\dots,
N_{i-2},M_{i-2},d,M_i,N_{i+1},M_{i+1}...$. By 
lemma \ref{lemmaTimeOnQuadri}, we have $d(\beta_2)<d(\beta)$. 
\end{proof}

\begin{prop}\label{reduction_to_imm}
  Let  $\beta \in \fxc^{r_\infty}\setminus    \fxc^{r_\infty,imm,0}$ and suppose
  that   $s(M_0)< s_{herd}$. Then there exists 
  $\beta'\in  \fxc^{r_\infty,imm,0}$ with $d(\beta')< d(\beta)$. 
\end{prop}
\begin{proof}
  Similar to the proof of proposition \ref{reduction_to_cross}. 
\end{proof}

We are now ready to prove Theorem
\ref{thm:optimal_constant_mitigations}.
\begin{proof}
  Item 1) is proved as in proposition \ref{pre-thm}.

  Items 6,7,8 are
  direct consequences of proposition \ref{reduction_to_cross}
  and \ref{reduction_to_imm}.

  For item 3, if $s(M_0)\geq s_{herd}$, it follows from
proposition \ref{reduction_to_cross} and proposition \ref{pre-thm}
  that $\inf_{\beta \in
    \fxc^{r_\infty}}d(\beta)=\min _{\beta \in
    \fxc^{r_\infty},0,cross}d(\beta)$.  If $s(M_0)< s_{herd}$,
proposition \ref{reduction_to_imm} shows that $\inf_{\beta \in
    \fxc^{r_\infty}}d(\beta)=\inf _{\beta' \in
    \fxc^{r_\infty},0,imm}d(\beta')$. But there is a unique
  $\beta' \in    \fxc^{{r_\infty},0,imm}$, so the infimum is
a minimum. 

Item 4 is easy.

  Let $D\subset [0,+\infty[$ be the set of elements $d$ such that
  $\fxc_{d}$ contains an optimal control $\beta$. Let $B\subset
  [0,1]$ be the set of elements $r_{\infty}$ such that
  $\fxc^{r_\infty}$ contains an optimal control.  By item 3 and 4,
  we have $B=]r_\infty(\beta_1),r_\infty(\beta_0)]$. 
  Let $\beta
  \in \fxc$. If $\beta$
  is optimal ( equivalently in  $\fxc^{r_\infty(\beta)}$ or
  in $\fxc_{d(\beta)}$ ), then $d(\beta)\in D$, $r_\infty
  (\beta) \in B$, and
  $d(r_\infty(\beta))=d(\beta)$ and
  $r_{\infty}(d(\beta))=r_\infty(\beta)$.
  This proves that the functions $r_{\infty}$ and $d$ in item 5) are
  mutually inverse one-to-one correspondences. If $s(M_0)\geq
  s_{herd}$, then $\beta\in \fxc^{r_{\infty}}$ is optimal implies that
  $\beta \in  \fxc^{r_{\infty},0,cross}$. Thus $d(\beta)$ has
  already been computed in proposition \ref{pre-thm}. In particular,
  we have seen in proposition  \ref{pre-thm} that $d(B)=[0,+\infty[$.
  This proves item 5 and since $D=d(B)=[0,+\infty[$, item 2 is proved
  too in the case $s(M_0)\geq
  s_{herd}$.  If $s(M_0)\leq s_{herd}$, the proof is similar using the analogous of
  proposition \ref{pre-thm}. The analogous statement consists in replacing $\fxc^{cross}$ with
  $\fxc^{imm}$. The proof of the analogous proposition is basically
  the same. The main change is that 
  item 2 is
  trivial when  $s(M_0)\leq s_{herd}$ since $\fxc_{d,0}^{imm}$ is a
  single point. 
\end{proof}

The following example is a confirmation of Theorem
\ref{thm:optimal_constant_mitigations}
in a
simple case where numeric computation is possible. It
considers the case of an absolute
lockdown, \textit{i.e.}. $\beta_1=0$.

\begin{ex}\label{confinementConstant}
  Let $\beta_1=0$. We consider the case $s(M_0)>s_{herd}$ and a one
  step mitigation of a fixed duration $d$. Then the optimum
  control minimising $r_\infty$ is in
  $\fxc^{cross}$. In other words, the value of $s$ is constant with
  $s=s_{herd}$
  during the mitigation. 
\end{ex}
\begin{proof}
When $\beta_1=0$ and $k=0$, the point $M(t)$
representing the epidemic starts at $t=0$ on the $R_0$-leaf
defined by $\ln s +R_0 r=c$ with $c=\ln s(M_0)+R_0r(M_0)$. At some time $t$,  we reach a point
$(s_1,r_1)$ with $\ln s_1+R_0r_1=c$
where the mitigation starts. At the end of the mitigation,
solving directly the sir system, we get
$M(t)=(s_1,r_2)$, with $r_2= r_1+i_1 (1-e^{-\mu d})$ and
$i_1=1-r_1-s_1$. To minimise $r_\infty$, we minimise the energy of $M(t)$, or equivalently
we maximise $H(s_1,r_2)= \ln s_1+R_0 r_2= \ln s_1 + R_0 r_1 + R_0
i_1(1-e^{-\mu d})=c+ R_0
i_1(1-e^{-\mu d})
$. The optimum is thus obtained when $i_1$ is maximal on the $R_0$-curve, thus
$s_1=s_{herd}$ by proposition \ref{prop:i_increasing_then_decreasing}.

%Then $H_0(s_0,r_0)= \ln s_0+R_0 r_1= \ln s_0 + R_0 r_0 + R_0 i_0(1-e^{-\mu d})$. But at the beginning, we are on the initial $R_0$-curve and therefore: 
 % $\ln s_0 +R_0 r_0=0$. Then $H_0(s_0,r_1)= R_0i_0 (1-e^{-\mu d})$ and the maximum is obtained when $i_0$ is maximal, in other words by the previous results, 
 % $s_0=s_{herd}$. 
\end{proof}

\subsection{Strategies with hospital saturation}

In this section, we consider situations where the health system 
becomes saturated when the epidemic evolves naturally from a
point $M_0$. Some mitigations are
triggered to avoid the saturation that would naturally occur. 
Mathematically, there is a share $i_{hosp}\in [0,1]$  of infected people
corresponding to a completely full but not overloaded health system. Without mitigations, the
maximum $i_{max}$ of $i(t)$ would  satisfy $i_{max}>i_{hosp}$ and the
system would be overloaded.  
To avoid saturation, a mitigation is triggered when some
level $i(t)=i_{trig}$ is reached, with $i_{trig}\leq i_{hosp}$.
The ratios $i_{hosp}$ and $i_{max}$ are given. The ratio $i_{trig}$
is chosen and this section discusses the choice of $i_{trig}$  to get
an efficient strategy. 

In the context of monitoring the charge of the health system,
some people propose to react sooner, while others
propose to wait longer before launching a mitigation. 
The first choice corresponds to
$i_{trig}<<i_{hosp}$ and the second choice to
$i_{trig}=i_{hosp}-\epsilon$
with $\epsilon\geq 0$ a small number.  Is it preferable to react
sooner or to wait ? What is the best value for $i_{trig}$ ? 

We consider two scenarios, the first one without rebound, the
second one allowing a rebound of the number of infected people
when the mitigation stops. In both scenarios, we suppose that
at $t=0$, the triggered
level is not  passed ($i(M_0)\leq i_{trig}$ ), and that after some time, 
the saturation would occur in the absence of  mitigation ($i_{max}>i_{hosp}$). This implies
in particular that $s(M_0)>s_{herd}$ (since $i$ is a decreasing
function of time if $s\leq s_{herd}$) and $i(M_0)>0$. We will assume that all these
assumptions hold in this section.  In summary $0<i(M_0)\leq
i_{trig}\leq i_{hosp}< i_{max}$, and $s(M_0)>s_{herd}$.  We will see
later that we can replace the condition $i_{trig}\in
[i(M_0),i_{hosp}]$ with $i_{trig}\in ]0,i_{hosp}]$ when $H(M_0)\geq
0$ with a suitable change of $M_0$
(Remark \ref{rem:low_itrig_allowed}).

\subsubsection{Scenario without rebound}

\begin{figure}[h]
\centering
\includegraphics[scale=.3]{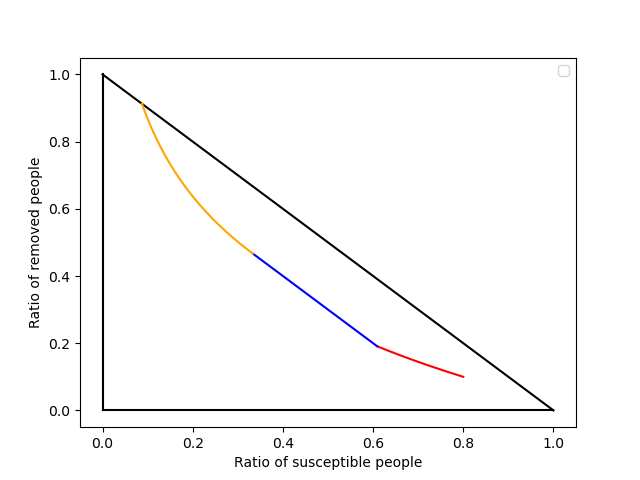}
\caption{Scenario with $M_0=(0.8,0.1)$, $i_{trig}=0.2$.}
\label{sansRebondM0NonTrivial}
\end{figure}

\begin{defi}
  The scenario without rebound starts at $t=0$ at a point $M_0$ with
  $i(M_0)>0$ and $s(M_0)>s_{herd}$. It is defined as follows. 
  \begin{itemize}
  \item The scenario depends on $i_{trig}\in [i(M_0),i_{hosp}]$ and we denote by
    $\beta_{i_{trig}}$ the corresponding control.
  \item The scenario is divided in three stages $t\in [0,t_{trig}]$,
    $t\in [t_{trig},t_{relax}]$, $t\geq t_{relax}$.
  \item For $t\leq t_{trig}$, the disease evolves with no constraint:
    $( \beta_{trig}(t)=\beta_0)$ and $i(t)\leq i_{trig}$.
\item For   $t \in [t_{trig},t_{relax}]$, the control $\beta_{trig}(t)$ (
  hence the mitigation policies)  is calibrated so that $i(t)=i_{trig}$
  is constant  on this period. 
  \item For $t\geq t_{relax}$, $\beta_{trig}(t)=\beta_0$ : all
    constraints are removed.
   \item $t_{trig}$ is defined as the smallest $t$ such that
    $i(t)=i_{trig}$ ( the mitigation starts when the critical level is
    reached). 
  \item $t_{relax}$ is characterised by $s(t_{relax})=s_{herd}$ ( the
    mitigation is removed when $i(t)$ decreases naturally, so that
    $i(t)$ will never exceed the critical level $i_{trig}$). 
    % In
    % particular, $h(t)=constant$ for $t> e_{hosp}$. We denote this
    % constant by $h_{\infty}$.
  \end{itemize}
\end{defi}

Since $M_0$ can be replaced with an other point $M_0'$ with smaller
$i$ without changing the problem, we can choose any $i_{trig}\in
]0,i_{hosp}]$. Let us formalise this remark. 
\begin{rem} \label{rem:low_itrig_allowed}
  We asked for the condition $i_{trig}\in [ i(M_0),i_{hosp}]$ with
  $i(M_0)>0$. In fact, we can consider $i_{trig}\in ]0,i_{hosp}]$ if
  $H(M_0)\geq 0$ using a suitable change of $M_0$. 
  
  Indeed, if $M'_0$ and $M_0$ are on the same $R_0$-curve with
  $s(M'_0)>s(M_0)$ and $i(M_0)>0$, then the  epidemic starting at
  $M'_0$ goes through $M_0$. It follows from the above description of
  the scenario that the mitigation is the
  same. The only difference between the two situations is that the initial
  part ( before the mitigation occurs, in red on figure \ref{sansRebondM0NonTrivial})
  is longer when the initial point is $M'_0$. In particular,
  when $i_{trig}$ is fixed, the cost and the $r_{\infty}$ of the
  strategy is the same if we replace $M_0$ by $M'_0$. 

  As a consequence the inequality $i_{trig}\geq i(M_0)$ is not
  necessary; the inequality $i_{trig}\geq i(M'_0)$ for some $M'_0$ as
  above is enough. If $H(M_0)\geq 0$, then  $i(M_{init})=0$ in Theorem
  \ref{thm:foliation}, item5. Thus  $i(M'_0)$ is arbitrarily small and
  the required condition on $i_{trig}$ is 
  $i_{trig}>inf(i(M'_0))=0$.

  Since $H(1,0)=0$ and since mitigations increase $H$, the condition
  $H(M_0)\geq 0$ is true if $M_0$ is a situation obtained
  after an epidemic has started and possibly some mitigations occurred. 
\end{rem}

\begin{thm}\label{thm:witoutRebound}
  Scenario without rebound.
  \begin{itemize}
  \item The share  $i(t)$ of infected people is increasing for $t\leq t_{trig}$, constant for
    $t\in [t_{trig},t_{relax}]$, decreasing for $t\geq t_{relax}$. 
  \item $R(t):=\frac{\beta(t)}{\mu}$ has constant value $R(t)=R_0$  for $t\leq t_{trig}$,
    is increasing continuously for
    $t\in [t_{trig},t_{relax}]$ from $\frac{1}{s(t_{trig})} $ to $R_0$ ,
    is constant equals to $R_0$ for $t\geq t_{relax}$. In particular,
    $R(t)$ is continuous except at $t=t_{trig}$.  
  \item $r_{\infty}$ is a strictly increasing function of the
    parameter $i_{trig}$. 
  \item If $i'_{trig}<i_{trig}$, then for every  cost function $c$,
    $c(\beta_{i'_{trig}})>c(\beta_{i_{trig}})$. In other words, a small
      $i_{trig}$ is universally costly.
    \item Suppose that both $i_{trig}$ and $M_0$ are varying. The duration $t_{relax}-t_{i_{trig}}$ of the mitigation tends to $+\infty$ when $i_{trig}
      $ tends to $0$ and $M_0$ stays in a region $s\geq s_{min}$ with
      $s_{min}>s_{herd}$. 
  \end{itemize}
\end{thm}
\begin{proof}
  First we remark that $t_{trig}$ is well
  defined. Indeed, by hypothesis $i_{max}>
  i_{hosp}\geq i_{trig} \geq i(M_0) \geq 0$.  Thus, without mitigation $i(t)$ would
  increase from $i(M_0)$ to $i_{max}$ and there is some $t$ with
  $i(t)=i_{trig}$ by intermediate value theorem.\\
  In the absence of mitigation,  $i_{max}$ is realised at a point $(s,r)$ with
  $s=s_{herd}$ and $s$ is a strictly decreasing
  function of $t$. Since the mitigation occurs before $i=i_{max}$,
  we have
  $s_{trig}>s_{herd}$. For $t\geq t_{trig}$, the trajectory is on the
  diagonal line $D$ with equation  $r+s=1-i_{trig}$. It is geometrically clear that $D$ intersects
the vertical line   $V$ with equation $s=s_{herd}$ in a point $N\in T$. Thus $t_{relax}$
  is well defined : it is defined by $t_{relax}-t_{trig}$ is the time
  spent on the segment $[M(t_{trig})N]$.
  \\
  Since on a $R_0$-leaf, $i(t)$ is increasing iff $s(t)\geq
  s_{herd}$ (proposition \ref{prop:i_increasing_then_decreasing}), the first item holds. 
  \\
  The second item follows from the formula for $\beta$ in Theorem
  \ref{thm:computationOfImplicitBeta}, with the remarks that $s'=-r'$
  when $i(t)$ is constant and that
  $s(t_{relax})=s_{herd}=\frac{1}{R_0}$.
  \\
  The number $r_{\infty}$ is an increasing function of the energy
  level $h_0(s,r)$ associated to a
  $R_0$-curve. For $t\geq t_{relax}$, $M(t)$ lies on a fixed $R_0$
  curve,  thus it suffices to compute the value of $h_0$ at time
  $t=t_{relax}$ to characterise $r_{\infty}$. Now
  $M(t_{relax})=(s_{herd},r_{relax}=1-s_{herd}-i_{relax}=1-s_{herd}-i_{trig})$. Since the
  first variable $s=s_{herd}$ is fixed, the formula for $h_0$ yields that $h_0$ is a
  monotonous function of $r_{relax}$, hence of $i_{trig}$. This proves
  the third item.   \\
  As for the fifth item, since $dt=\frac{dr}{\mu i} $ by the
  sir equations,  we have
  $t_{relax}-t_{trig}=\int_{M(t_{trig})}^{M(t_{relax})}\frac{dr}{\mu i}
  =\frac{1}{\mu i_{trig}}
  \int_{M(t_{trig})}^{M(t_{relax})}dr=\frac{r(t_{relax})-r(t_{trig})}{\mu
    i_{trig}} $. When $i_{trig}(n)$ is a sequence that  tends
  to $0$, and $M_0(n)$ is a sequence of initial points that stays in
  the zone $[s_{min},1[$, then the limit ( as a function of $n$ ) of
  $r(t_{relax}(n))=1-s(t_{relax}(n))-i_{trig}(n)=1-s_{herd}-i_{trig}(n)$ is
  $1-s_{herd}$. Since $i_{trig}(n)$ tends to $0$, and since in the
  zone $s>s_{min}$, a $R_0$-leaf has a tangent with slope
  $\frac{dr}{ds} =-\frac{1}{ R_0s}\geq \frac{-1}{ R_0s_{min}}>-1$, the distance between
  $M_0(n)$ and $M(t_{trig}(n))$ tends to $0$. In particular, $\limsup
  r(t_{trig}(n))= \limsup r(M_0(n))=1-\liminf s(M_0(n))\leq 1-s_{min}<
  1-s_{herd}$. The formula for $t_{relax}-t_{trig}$
  then implies that $\lim_{n\to +\infty}
  t_{relax}(n)-t_{trig}(n)=+\infty$. \\
  The trajectory triggering at the level $i_{trig}$ is constrained on
  the segment $[M(t_{trig})),M(t_{relax})]$. This segment is included
  in a line $L:i=cte$. The same remark for $i'_{trig}$ corresponds to
  the segment $[M'(t'_{trig}),M'(t'_{relax})]$ and to a line $L'$. We consider $\tilde \phi:L\to L'$
  the geometrical affine map which sends $M(t_{trig})$ to
  $M'(t'_{trig})$ and $M(t_{relax})$ to  $M'(t'_{relax})$. We consider
  $\phi$ the induced temporal map $[t_{trig},t_{relax}] \to
  [t'_{trig},t'_{relax}]$ defined by 
  $\tilde \phi(M(t))=M'(\phi(t))$. For simplicity, we use the notations
  $t'=\phi(t)$, $\beta_{i_{trig}}=\beta$ and $\beta_{i'_{trig}}=\beta'$. By theorem
  \ref{thm:dominatedImpliesInferior}, to settle the fourth item,
  we need to prove that $\beta'(t')\leq \beta(t)$ and that
  $\frac{d\phi}{dt}\geq 1$. We have $\beta(t)=\frac{\mu}{s(M(t))} $ and
  $\beta'(t')=\frac{\mu}{s(M'(t'))} $ by Theorem
  \ref{thm:computationOfImplicitBeta}. The inequality $s(M'(t'))>s(M(t))$ is
  clear geometrically and easy to prove. The map $\phi$ sends
  $[t,t+dt]$ to $[t',t'+\frac{d\phi}{dt} dt=t'+dt']$. We want
  $dt'>dt$. By the sir equations, $dt=\frac{dr}{\mu i_{trig}} $ and
  $dt'=\frac{dr'}{\mu i'_{trig}} $. Since $\tilde \phi$ is an affine
  dilatation, we have $dr'>dr$. Finally since $i'_{trig}<i_{trig}$, it
  follows that $dt'>dt$.   
\end{proof}

\begin{coro}
  For all $i_{trig}$, $j_{trig}\in ]0,i_{hosp}[$, the corresponding
  controls $\beta_{i_{trig}}$ and $\beta_{j_{trig}}$
  are not comparable, thus no strategy is preferable. 
\end{coro}
\begin{proof}
  By Theorem \ref{thm:witoutRebound},  if $i_{trig}<j_{trig}$, $r_{\infty}$ is lower for $i_{trig}$ but
   at the price of a more costly mitigation. 
\end{proof}

\subsubsection{Scenario with rebound}
The only difference between the scenario with rebound and the scenario
without rebound is on the choice of $t_{relax}$. In the previous
scenario without rebound,
the end of the mitigation was late, $t_{relax}$ was chosen large so that $i(t)$ decreases
for  $t\geq t_{relax}$. In the scenario with rebound, the mitigation is
relaxed sooner. As a consequence,  a rebound of $i(t)$ occurs when the mitigation
stops.
However, the rebound must be small enough not to overload
the health system. In technical terms, the inequality $i(t)\leq i_{hosp}$ must remain
true for $t\geq t_{relax}$. This property implicitly defines $t_{relax}$ :
The mitigation is relaxed as soon as possible provided 
the health system is not overwhelmed by the rebound.
By construction, there are less constraints for this scenario with
rebound in comparison to the scenario without rebound, since this
is the same level of constraints, but relaxed earlier. 
The formal definition of the strategy with rebound
is given in the next definition. 

\begin{figure}[h]
\centering
\includegraphics[scale=.3]{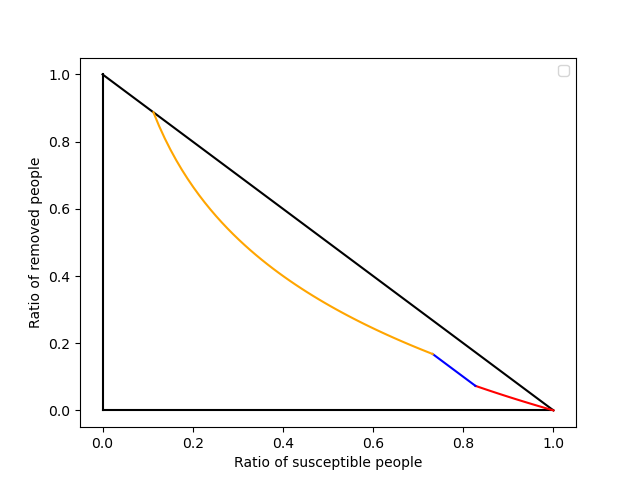}
\caption{Scenario with rebound, $M_0=(0.999,0)$.}
\label{illustrationAvecRebond}
\end{figure}

\begin{defi} The scenario with rebound starts at $t=0$ at the point
  $M_0\in T$, with $i(M_0)>0$, $s(M_0)>s_{herd}$. It
  depends on the parameter
  $i_{trig}\in [i(M_0),i_{hosp}]$ and it is defined via its
  control $\beta_{i_{trig}}$ as follows. 
  \begin{itemize}
  \item The scenario is divided in three stages $t\in [0,t_{trig}]$,
    $t\in [t_{trig},t_{relax}]$, $t\geq t_{relax}$.
  \item For $t\leq t_{trig}$, the disease evolves with no constraint:
    $( \beta_{trig}(t)=\beta_0)$ and $i(t)\leq i_{trig}$.
  \item $t_{trig}$ is defined as the smallest $t$ such that
    $i(t)=i_{trig}$.
  \item The mitigation policies are adjusted so that $i(t)=i_{trig}$
    for $t \in [t_{trig},t_{relax}]$
  \item For $t\geq t_{relax}$, $\beta_{trig}(t)=\beta_0$ : all
    constraints are removed.
  \item $t_{relax}$ satisfies $t_{relax}\geq t_{trig}$ and it is the
    smallest such $t$ satisfying $i([t,+\infty[)\subset
    [0,i_{hosp}]$ in the absence of mitigation after $t$.
    In other words, the health system is never
    overwhelmed after $t_{relax}$.
    % In
    % particular, $h(t)=constant$ for $t> e_{hosp}$. We denote this
    % constant by $h_{\infty}$.
  \end{itemize}

\end{defi}

Remark \ref{rem:low_itrig_allowed} applies here and we will
can replace the condition  $i_{trig}\in [i(M_0),i_{hosp}]$ 
with $i_{trig}\in ]0,i_{hosp}]$ if $H(M_0)>0$. 

\begin{lm}
  Let $C_{hosp}$ be the $R_0$-curve containing the point
  $(s,r)=(s_{herd},1-i_{hosp}-s_{herd})$. Let $\Delta_{i_{trig}}$ be the
  line $1-r-s=i_{trig}$. The intersection $C_{hosp}\cap
  \Delta_{i_{trig}}$ contains a point $M_r$ satisfying
  $s_{herd}<s(M_r)$.  Moreover, $M(t_{relax})=M_r$.
  In particular, $r_{\infty}(\beta_{i_{trig}})$ is independent of
  $i_{trig}$ as for $t$ large enough, the point $M(t)$  is on the $R_0$-curve
  $C_{hosp}$. 
\end{lm}
\begin{proof}
  We have parameterised any $R_0$-leaf by a constant $c$.
  The $R_0$-leaf with constant $c$ has equation $\ln(s)+R_0r=c$.
  Let $i_{carac}:=1-s_{herd}-\frac{c-\ln(s_{herd})}{R_0}
  $. Geometrically, the meaning of $i_{carac}$ is the following. If the $R_0$-leaf
  intersects the vertical line $s=s_{herd}$ in a point $M$, $i_{carac}$ is the 
  value of $i(M)$.

  Note that any $R_0$-leaf is characterised by the
  value of $i_{carac}$. The $R_0$-leaf passing through $M(0)=M_0$ has $i_{carac}=i_{max}$ by definition
  of $i_{max}$ and proposition
\ref{prop:i_increasing_then_decreasing}.
The $R_0$-curve $C_{hosp}$ is defined by $i_{carac}=i_{hosp}$. 

  We extend the definition of $i_{carac}$ from $R_0$-leafs to points $M\in T$ : we let
  $i_{carac}(M)=i_{carac}(C)$ where $C$ is the $R_0$-curve through
  $M$ ( In formula :  $i_{carac}(M)=1-s_{herd}-\frac{\ln(s(M))+R_0r(M)-\ln(s_{herd})}{R_0} $). 
  
The two points $M(t_{trig})$ and $N=(s_{herd},1-s_{herd}-i_{trig})$
are on $\Delta_{i_{trig}}\cap T$. We have $i_{carac}( M(t_{trig}))=
i_{carac}(M(0))=i_{max}$ and $i_{carac}(N)=i_{trig}$.
 Since $i_{trig}\leq i_{hosp} \leq i_{max}$,  by
  the intermediate value theorem, there exists $M_r \in
  [M(t_{trig}),N]$ with $i_{carac}(M_r)=i_{hosp}$ (
  \textit{i.e.}. $M_r\in C_{hosp}$).

By definition of the strategy, the point $M(t_{relax})$ is a point on
the segment $[M(t_{trig}),N]$. If $ M(t_{relax})\in
[M(t_{trig}),M_r[$, then $i_{carac}(M(t_{relax}))>i_{hosp}$ is too large and the
health system is overwhelmed.  If $ M(t_{relax})\in ]M_r,N]$, then
 $i_{carac}(M(t_{relax}))<i_{hosp}$ and it is possible to reduce the
 duration of the mitigation with no overload on the health system, contradicting
 the minimality of $t_{relax}$.  Thus $M(t_{relax})=M_r$.

  %   Remark that, $C(i'_{carac})$ is the translation of $C(i_{carac})$ by
  % the vector $(0,i_{carac}-i'_{carac})$.  The initial $R_0$-curve is
  % $C(i_{max})$. The curve $C_{hosp}$ in the statement is reformulated
  % as $C_{i_{hosp}}$. Since $i_{hosp}<i_{max}$ by saturation
  % hypothesis, the curve $C_{hosp}$ lies above the initial curve
  % $C_{i_{max}}$.

  % We let $s_{trig}=s(M(t_{trig}))$. For any $R_0$-curve $C$, we define the function
  % $i_{C}(s)$ as follows. There exists a unique $r$ such that $r,s\in
  % C$, and we let $i_C(s)=1-s-r$. For the curve $C_{hosp}$, we have
  % $i_{C_{hosp}}(s_{herd})=i_{hosp}$ and
  % $i_{C_{hosp}}(s_{trig})<i_{C_{i_{max}}}(s_{trig})=i_{trig}$. Thus by
  % the intermediate value theorem, there exists $s\in
  % ]s_{herd},s_{trig}[$ such that $i_{C_{hosp}}(s)=i_{trig}$. We put
  % $M_{r}=(s,1-s-i_{trig})$. \\
  % All the $R_0$ curves below  $C_{hosp}$ have $i_{carac}>i_{hosp}$. In
  % particular, $M(t_{relax})$ cannot be a point in the open segment
  % $[M_{trig},M_r[$ otherwise the hospital capacity would be
  % overwhelmed. Moreover, $M(t_{relax})$ cannot be above $C_{hosp}$
  % otherwise $t_{relax}$ is not minimal. Thus $M(t_{relax})$ lies on
  % $C_{hosp}$ and is equal to $M_r$. 
\end{proof}

\begin{thm}\label{thm:withRebound}
  Scenario with rebound.
  \begin{enumerate}
  \item Let $t_{herd}$ be the time such that $s(t_{herd})=s_{herd}$. Then
    $t_{herd}>t_{relax}$. 
    The quantity $i(t)$ is increasing for $t\leq t_{trig}$, constant for
    $t\in [t_{trig},t_{relax}]$, increasing for
    $t\in[t_{relax},t_{herd}]$, decreasing for $t\geq t_{herd}$. 
  \item $R(t):=\frac{\beta(t)}{\mu}=R_0$  for $t\leq t_{trig}$,
    increasing continuously for
    $t\in [t_{trig},t_{relax}]$ from $\frac{1}{s(t_{trig})} $ to $\frac{1}{s(t_{relax})}$ ,
    constant equals to $R_0$ for $t\geq t_{relax}$. In particular,
    $R(t)$ is continuous except at $t=t_{trig}$ and $t_{relax}$.  
  \item The maximum of constraint is
    $infess(\beta)=\frac{\mu}{s(t_{trig})}$. It is an increasing
    function of $i_{trig}$. Thus, from this point of view, the mitigation is harsher for a small
    $i_{trig}$. 
  \item $r_{\infty}$ is constant independent of the
    parameter $i_{trig}$ thus the strategies are ordered only by the
    cost functions. 
  \item Suppose $H(M_0)\geq 0$. There exists a unique $i_{min}$ such that the duration
    $d(i_{trig}):=t_{relax}-t_{i_{trig}}$ of the
    mitigation as a function of $i_{trig}$ is decreasing for
    $i_{trig}\in ]0,i_{min}]$, increasing for $i_{trig}\in
    ]i_{min},i_{hosp}]$.
  \item     If $i_{trig}<j_{trig}\leq i_{min}$, then for every  cost function $c$,
    $c(\beta_{i_{trig}})>c(\beta_{j_{trig}})$. In other words, a small
      $i_{trig}$ is universally more costly.
%    \item The duration $d(i_{trig})$ of the mitigation tends to $+\infty$ when $i_{trig}
%      $ tends to $0$. 
  \end{enumerate}
\end{thm}
\begin{proof}
  The first three items are proved like in the scenario without
  rebound ( Theorem \ref{thm:witoutRebound}).
  The fourth item is a direct consequence of the previous
  lemma.

  We now prove item 5. Since $H(M_0)\geq 0$, we suppose that
  $i_{trig}\in ]0,i_{hosp}]$ by remark \ref{rem:low_itrig_allowed}.
  \\
    We define $s_1,s_2,r_1,r_2$ functions of $i_{trig}$ by
  $(s_1(i_{trig}),r_1(i_{trig}))=M(t_{trig})$ and
  $(s_2(i_{trig}),r_2(i_{trig}))=M(t_{relax})$.
  By construction, for $k=1$ or $2$, 
 $r_k+s_k= 1-i_{trig}$ and $R_0 r_k + \ln s_k=H_k$ is a constant
 independent of $i_{trig}$.  By derivation with respect to $i_{trig},$ it follows: 
 $r_{k}'+s_{k}'=-1$,  $R_0 r_{k}'+ \frac{s_{k}'}{s_k}=0$,
 and $s_{k}' = \frac{R_0 s_k}{1-R_0 s_k}$.

 We have $d(i_{trig})= \int_{M(t_{trig})}^{M(t_{relax})}dt=
 \frac{r_2(i_{trig})-r_1(i_{trig})}{\mu i_{trig}}$, which follows from
 the sir-relation $\frac{dr}{dt} =\mu i_{trig}$. 
 Thus $d'$ has the sign of $i_{trig}(r_{2}'-r_{1}')-(r_2-r_1)=
 i_{trig} ( \frac{R_0 s_1}{1-R_0s_1}-\frac{R_0 s_2}{1-R_0 s_2}
 )-(s_1-s_2)=(s_1-s_2)   \left[
   \frac{i_{trig}R_0}{(1-R_0s_1)(1-R_0s_2)}-1 \right]$.  Since
 $s_1-s_2$, $1-R_0s_1$ and $1-R_0s_2$ are positive, the sign of $d'$ is the sign of the function
 (of $i_{trig})$ $W= i_{trig}R_0- (R_0 s_1-1)(R_0s_2-1)$.  Now we make the following remarks: 

 \begin{itemize}
\item $W(i_{trig})$ is a strictly increasing function of $i_{trig}\in
  ]0,i_{hosp}]$.  This comes from
  the fact that  $i_{trig}R_0$  is strictly increasing and that $s_1$ and $s_2$ are decreasing.

\item $s_2(i_{hosp})= \frac{1}{R_0}$ and then $W(i_{hosp})= i_{hosp}R_0>0$. 

\item  The functions $s_1,s_2$ have limits $s_1(0)$ and $s_2(0)$ when
 $i_{trig}$ tends to $0$. We have $s_1(0)>\frac{1}{R_0} $ and
 $s_2(0)>\frac{1}{R_0} $. Thus the limit
 $W(0)=-(R_0s_1(0)-1)(R_0s_2(0)-1)$ of $W$ is strictly negative.

\end{itemize}
It follows by the intermediate value theorem that $d'(i_{min})=0$ for
some $i_{min}\in ]0,i_{hosp}[$, and that 
$d'(]0,i_{min}[)\subset ]-\infty,0[$, $d'(]i_{min},i_{hosp}])\subset
]0,+\infty[$. This implies the fifth item.

  The last item is proved like in Theorem
  \ref{thm:witoutRebound}, with the following change
  to prove that $\frac{d\phi}{dt} \geq 1$, \textit{i.e.} that $\phi$ is a
  local dilatation of the time. The map $\phi$
  is affine in $r$, and since $i$ is constant, the sir equation
  $\frac{dr}{dt} =\mu i$ implies that $\phi$ is affine in $t$.
  It follows that $\phi$ is locally a dilatation of the time $t$ if and only if it is
  globally a dilatation. Item $5$ proves that $\phi$ is
  globally a dilatation and concludes the proof of item 6.

\end{proof}

\begin{coro}
  For all $i_{trig}<j_{trig}\in ]0, i_{min}]$,  the
  control $\beta_{j_{trig}}$  is preferable to the control $\beta_{i_{trig}}$. In particular, all
  strategies with $i_{trig}<i_{min}$ should be avoided. \\ 
  For all $i_{trig},j_{trig}\in [i_{min},i_{hosp}]$, $\beta_{i_{trig}}$ and $\beta_{j_{trig}}$
  are not comparable, thus no strategy is preferable. 
\end{coro}
\begin{proof}
  All the strategies with rebound have the same $r_{\infty}$, thus
  they are ordered by the cost functions. For $i_{trig}<i_{min}$, a
  small $i_{trig}$ corresponds to a universally high cost by the
  theorem.   
  If $i_{trig}\in [i_{min},i_{hosp}]$, then a small $i_{trig}$ corresponds to
  a short duration by item 5 but a high maximum constraint by item 3.
  Two strategies associated to
$i_{trig},j_{trig}\in [i_{min},i_{hosp}]$
  are then incomparable by Theorem \ref{thm:criterion_for_maximal_control} 
\end{proof}

In summary,  the theorem tells that if the
mitigation occurs when the health system is insufficiently full,
$\beta_{i_{trig}}$  is not a maximal control and other
controls are preferable. 
For the choice of
$i_{trig}$ to be
sound, it is necessary that $i_{trig}\geq i_{min}$.
Once $i_{trig}\geq i_{min}$, all choices make sense  and
compromises occur : a 
larger $i_{trig}$
corresponds to interventions that last longer with softer maximal constraint. In
particular, the control $\beta_{i_{hosp}}$ has the longest constraint
time but the softest maximal constraint.

\bibliographystyle{plain} 

\end{document}